\documentclass[10pt, final, journal, letterpaper, twocolumn]{IEEEtran}
\makeatletter
\def\ps@headings{%
\def\@oddhead{\mbox{}\scriptsize\rightmark \hfil \thepage}%
\def\@evenhead{\scriptsize\thepage \hfil \leftmark\mbox{}}%
\def\@oddfoot{}%
\def\@evenfoot{}}
\makeatother 

\IEEEoverridecommandlockouts
\usepackage{subfigure}
\usepackage{bbm}
\usepackage{amsfonts}
\usepackage[dvips]{graphicx}
\usepackage{times}
\usepackage{cite}
\usepackage{amsmath}
\usepackage{array}
\usepackage{amssymb}

\newcommand{\bs}{\boldsymbol}

\usepackage{stfloats}
\usepackage{slashbox}
\usepackage{graphicx}
\usepackage{footnote}
\usepackage{booktabs}
\usepackage{array}
\usepackage{algorithm}
\usepackage{subeqnarray}
\usepackage{cases}
\usepackage{threeparttable}
\usepackage{color}
\usepackage{hyperref}
\usepackage{epstopdf}
\usepackage{algpseudocode}
\usepackage{bm}
\usepackage{multirow}
\usepackage{adjustbox}

\newtheorem{proposition}{Proposition}

\begin{document}

\title{Joint Design of Hybrid Beamforming and Reflection Coefficients in RIS-aided mmWave MIMO Systems}

\author{\authorblockN{Renwang Li, Bei Guo, Meixia Tao,\IEEEmembership{ Fellow,~IEEE}, Ya-Feng Liu,\IEEEmembership{ Senior Member,~IEEE}, and Wei Yu,\IEEEmembership{ Fellow,~IEEE}}
\thanks{Part of this work was presented at IEEE Wireless Communications and Networking Conference (WCNC) 2021 \cite{9417417} [DOI: 10.1109/WCNC49053.2021.9417417]. The work of R. Li, B. Guo and M. Tao was supported in part by the National Natural Science Foundation of China (NSFC) under Grant 61941106 and Grant 62125108. The work of Y.-F. Liu was supported in part by NSFC under Grant 12021001 and Grant 11991021. The work of Wei Yu was supported by the Canada Research Chairs program. (Corresponding author: Meixia Tao.)}
\thanks{R. Li, B. Guo and M. Tao are with Department of Electronic Engineering, Shanghai Jiao Tong University, Shanghai, China (emails:\{renwanglee, guobei132, mxtao \}@sjtu.edu.cn).}
\thanks{Y.-F. Liu is with the State Key Laboratory of Scientific and Engineering Computing, Institute of Computational Mathematics
and Scientific/Engineering Computing, Academy of Mathematics and Systems Science, Chinese Academy of Sciences, Beijing
100190, China (e-mail: yafliu@lsec.cc.ac.cn).}
\thanks{W. Yu is with Department of Electrical and Engineering, University of Toronto, Toronto, ON, Canada, M5S 3G4 (e-mail: weiyu@ece.utoronto.ca).}
}
\maketitle
\begin{abstract}
  This paper considers a reconfigurable intelligent surface (RIS)-aided millimeter wave (mmWave) downlink communication system where hybrid analog-digital beamforming is employed at the base station (BS). We formulate a power minimization problem by jointly optimizing hybrid beamforming at the BS and the response matrix at the RIS, under the signal-to-interference-plus-noise ratio (SINR) constraints at all users. The problem is highly challenging to solve due to the non-convex SINR constraints as well as the unit-modulus phase shift constraints for both the RIS reflection coefficients and the analog beamformer. A two-layer penalty-based algorithm is proposed to decouple variables in SINR constraints, and manifold optimization is adopted to handle the non-convex unit-modulus constraints. {We also propose a low-complexity sequential optimization method, which optimizes the RIS reflection coefficients, the analog beamformer, and the digital beamformer sequentially without iteration.} Furthermore, the relationship between the power minimization problem and the max-min fairness (MMF) problem is discussed. Simulation results show that the proposed penalty-based algorithm outperforms the state-of-the-art semidefinite relaxation (SDR)-based algorithm. Results also demonstrate that the RIS plays an important role in the power reduction.
\end{abstract}

\section{Introduction}
The millimeter wave (mmWave) communication over 30-300 GHz spectrum is a key technology in 5G and beyond wireless networks to provide high data-rate transmission \cite{6732923,6824746,niu2015survey}. Compared with sub-6 GHz, the high directivity at high frequency bands makes mmWave communication much more sensitive to signal blockage. One promising and cost-effective solution to overcome the blockage issue is to deploy Reconfigurable Intelligent Surfaces (RISs). An RIS is an artificial meta-surface consisting of a large number of passive reflection elements that can be programmed to electronically control the phase of the incident electromagnetic waves \cite{book,8910627}. With the help of a smart controller, RISs can be controlled to enhance the desirable signals via coherent combining, or to suppress the undesirable interference via destructive combining. RISs are spectrum- and energy-efficient since they do not require radio frequency (RF) components or dedicated energy supply. Furthermore, from the implementation perspective, RISs have appealing advantages such as low profile, light-weight, and conformal geometry. Recently, RISs have emerged as a promising technique to enhance the performance of wireless communication systems, especially in mmWave bands \cite{8796365,9122596,9086766}.

As RISs bring a new degree-of-freedom to the optimization
of beamforming design, a key issue of interest in RIS-aided wireless communication systems is to jointly design the active beamforming at the multi-antenna base stations (BSs) and the passive reflection coefficients at the RIS. There have been several prior studies investigating this problem under different system setups and assumptions \cite{ wu2019intelligent,9226616,9246254, 8741198, li2019joint, guo2020weighted}. Specifically, the work \cite{wu2019intelligent} studies the power minimization problem under the signal-to-interference-plus-noise ratio (SINR) constraints and proposes a semidefinite relaxation (SDR) based algorithm for the joint active and passive beamforming design. The work \cite{9226616} extends \cite{wu2019intelligent} to the scenario with multiple RISs and a near-optimal analytical solution is derived. The work \cite{9246254} aims to maximize the minimum weighted SINR at the users and proposes a low-complexity inexact-alternating-optimization approach. The work \cite{8741198} focuses on the energy efficiency problem under individual quality-of-service (QoS) requirements as well as maximum power constraints. Under the maximum transmit power constraints, the work \cite{li2019joint} aims to maximize the minimum SINR, and the work \cite{guo2020weighted} aims to maximize the weighted-sum-rate (WSR) of all users.
Moreover, RISs have also been studied under other communication setups, such as secure communication \cite{8723525, 9198898}, unmanned aerial vehicle (UAV) communication \cite{8959174, 9124704}, and simultaneous wireless information and power transfer (SWIPT) systems \cite{8941080, wu2020joint}.
Note that in all these works on joint active-passive beamforming design, the active beamforming at the BS is fully digital as in most of the multiple-input-multiple-output (MIMO) beamforming literature, which requires each antenna to be connected to one RF chain, and hence has a high hardware cost.

{Unlike the fully digital beamforming structure, hybrid analog and digital (A/D) beamforming at the BS is more practical in mmWave systems since it employs a reduced number of RF chains \cite{8030501,7389996}. It is therefore desirable to consider hybrid beamforming in RIS-aided mmWave communications as a cost-effective alternative. There are very few works along this line of research except \cite{ying2020gmdbased, xiu2020reconfigurable, 9234098}. In specific, the work \cite{ying2020gmdbased} considers the individual design of the digital beamformer, the analog beamformer, and the RIS phase shifts to achieve low error rate in a wideband system. The work \cite{xiu2020reconfigurable} investigates the WSR maximization in a nonorthogonal multiple access (NOMA) system by jointly designing the power allocation, the RIS phase shifts and the hybrid beamforming vector. Therein, the manifold optimization method is adopted for the design of the phase shifts at both the RIS and the analog beamformer, while the digital beamforming is obtained by the successive convex approximation (SCA) based algorithm. The work \cite{9234098} focuses on maximizing the spectral efficiency in a single-user mmWave MIMO system by jointly optimizing the RIS reflection coefficients and the hybrid beamforming vector at the BS. The manifold optimization is adopted to handle the RIS reflection coefficients, and then the digital beamforming is obtained through the singular value decomposition (SVD) of the cascaded channel.}

In this work, we consider an RIS-aided multi-user downlink mmWave system, and investigate the joint design of hybrid beamforming at the BS and reflection coefficients at the RIS.  Unlike the previous works \cite{ying2020gmdbased, xiu2020reconfigurable, 9234098} which all employ the fully-connected hybrid architecture at the BS with each RF chain connected to all antenna elements, we employ the sub-connected hybrid architecture with each RF chain only connected to a disjoint subset of antenna elements. The sub-connected architecture is more appealing for its further reduced hardware cost and power consumption.

The main contributions and results of this paper are listed as follows.
\begin{itemize}
\item {We first formulate the so-called QoS problem for minimizing the total transmit power at the BS subject to individual SINR constraints at all users. The problem is highly non-convex due to the deeply coupled variables and the unit-modulus phase shifts constraints. To tackle this problem, we propose a two-layer penalty-based  algorithm where the block coordinate descent (BCD) method is adopted in the inner layer to solve a penalized problem and the penalty factor is updated in the outer layer until convergence. The penalty method can decouple the optimization variables and make the problem much easier to handle. In the BCD method, considering the same unit-modulus constraints on both the BS analog beamformer and the RIS response matrix, they can be updated simultaneously by using the manifold optimization method.}
\item {In order to reduce the complexity, we propose a sequential optimization method where the RIS coefficients are obtained by maximizing the channel gain of the user with the worst channel state; the analog beamforming is obtained by minimizing the Euclidean distance between the fully digital beamforming and the hybrid beamforming; and the digital beamforming is optimally obtained by the second-order cone program (SOCP) method.}
\item {We discuss a closely related problem of the QoS problem, which is the max-min fairness (MMF) problem. The MMF problem is more difficult to solve than the QoS problem due to its non-smooth objective function. However, we can solve the MMF problem by solving a series of QoS problems.}
\end{itemize}

{Finally, we conduct comprehensive simulations to validate the performance of the proposed algorithms. It is shown that the proposed penalty-based algorithm outperforms the traditional SDR-based optimization algorithm. Results also demonstrate that the proposed hybrid beamforming at the BS can perform closely to a fully digital beamforming system. In addition, the transmit power at the BS can be greatly reduced by employing a large number of RIS elements on the BS side or the user side. Furthermore, it is sufficient for practical use when both the RIS and the analog beamformer have 3-bit quantizers.}

The rest of the paper is organized as follows. Section \ref{sec_model} introduces the RIS-aided mmWave MIMO system model, and formulates the power minimization problem.  A two-layer penalty-based algorithm is proposed to solve the power minimization problem in Section \ref{sec_qos}. A low-complexity sequential optimization method is proposed in Section \ref{individual}. The relationship between the QoS problem and the MMF problem is studied in Section \ref{discussion}. Simulation results are provided in Section \ref{sec_simulation}. Finally, Section \ref{sec_conclusion} concludes this paper.

\emph{Notations}: The imaginary unit is denoted by $j=\sqrt{-1}$. Vectors and matrices are denoted by bold-face lower-case and upper-case letters, respectively.
The conjugate, transpose, conjugate transpose and pseudo-inverse of the vector $\bf x$ are denoted by $ \bf x^*$, $\mathbf{x}^T$, $\mathbf x^H$ and $\bf x^\dagger$, respectively. Further, we use $\bf I$ and $\mathbf{O}$ to denote an identity matrix and all-zero matrix of appropriate dimensions, respectively; we use $\mathbb{C}^{x\times y}$ to denote the space of $x\times y$ complex-valued matrices. The notations  $\arg(\cdot)$ and $\operatorname{Re}(\cdot)$  denote the argument and real part of a complex number, respectively. The notations $\mathbb{E}(\cdot)$ and $\operatorname{Tr}(\cdot)$ denote the expectation and trace operation, respectively; $\odot$ represents the Hadamard product; $\|\cdot\|$ represents the Frobenius norm. For a vector $\mathbf{x}$, $\operatorname{diag}(\mathbf{x})$ denotes a diagonal matrix with each diagonal element being the corresponding element in $\mathbf{x}$. For a vector $\mathbf{x}$, $\nabla f(\mathbf{x}_i)$ denotes the gradient vector of function $f(\mathbf{x})$ at the point $\mathbf{x}_i$. Finally, The distribution of a circularly symmetric complex Gaussian (CSCG) random vector with mean vector $x$ and covariance matrix $\Sigma$ is denoted by $\mathcal{C}\mathcal{N}(x,\Sigma)$; and $\sim$ stands for ``distributed as''.

\section{System Model And Problem Formulation}
\label{sec_model}
\subsection{System Model}
\begin{figure}[t]
\begin{centering}
\vspace{-0.1cm}
\includegraphics[width=0.45\textwidth]{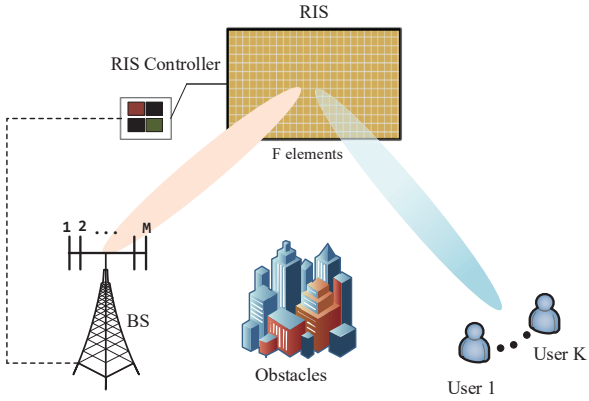}
 \caption{\small{An RIS-aided downlink mmWave communication system.}}\label{fig_systemmodel}
\end{centering}
\end{figure}
As shown in Fig.~\ref{fig_systemmodel}, we consider an RIS-aided downlink mmWave communication system where one BS, equipped with $M$ antennas, communicates with $K$ single-antenna users via the help of one RIS equipped with $F$ unit cells. The BS employs the sub-connected hybrid A/D beamforming structure with $N$ RF chains, each connected to a disjoint subset of $D=M/N$ antennas. Let $s_k$ denote the information signal intended to user $k$, for $k\in \mathcal{K} \triangleq \{1, 2, \ldots,K\}$. The signals are assumed to be independent of each other and satisfy $\mathbb{E}(|s_k|^2)=1$. Each of these signals is first weighted by a digital beamforming vector, denoted as $\mathbf{w}_k\in\mathbb{C}^{N\times{1}}$. These weighted signal vectors are summed together and each entry is sent to an RF chain, then multiplied by an analog beamforming vector, denoted as $\mathbf{v}_n\in \mathbb{C}^{D\times{1}}$, for $n\in \mathcal{N}\triangleq \{1,2,\ldots,N\}$. Each entry of $\mathbf{v}_n$, denoted as ${v}_{n,d},  \forall d\in \mathcal{D}\triangleq\{1,2,\ldots,D\}$ is a phase shifter. Discrete phase shifts are considered. Denote $\mathcal{S}_a$ as the set of all possible phase shifts for the analog beamformer, given by
\begin{equation} 
\mathcal{S}_{a} \triangleq \left\{e^{j\theta} \; \bigg| \;  \theta \in \left\{0,\frac{2\pi}{2^{Q_1}}, \ldots, \frac{2\pi(2^{Q_1}-1)}{2^{Q_1}}\right\} \right\},
\end{equation}
where $Q_1$ is the number of control bits for each analog phase shifter. In the special case when $Q_1 = \infty$, it becomes continuous phase shifts.  The overall analog beamforming matrix can be represented as
\begin{equation}
\mathbf{V}=\left[\begin{array}{cccc}
\mathbf{v}_{1} & \mathbf{0} & \cdots & \mathbf{0} \\
\mathbf{0} & \mathbf{v}_{2} & \cdots & \mathbf{0} \\
\mathbf{0} & \mathbf{0} & \ddots & \mathbf{0} \\
\mathbf{0} & \mathbf{0} & \cdots & \mathbf{v}_{N}
\end{array}\right].
\end{equation}
The total transmit power of the BS is given by
\begin{equation}
    P_\text{transmit}=\sum\limits_{k=1}^{K}\Arrowvert{\mathbf{V}\mathbf{w}_k}\Arrowvert^2
     =D \sum\limits_{k=1}^K  \Arrowvert{\mathbf{w}_k}\Arrowvert^2.
\end{equation}

The RIS is connected to the BS through an RIS control link for transmission and information exchange. Let $\mathcal{F} \triangleq \{1,2,\ldots,F\}$ denote the set of all RIS unit cells, and define the response matrix at the RIS as
\begin{equation}
  \mathbf{\Theta}= \text{diag} {(b_1,b_2,\ldots,b_F)},
\end{equation}
where $b_f=\beta_f e^{j\theta_f}, \beta_f \in [0,1]$ and $\theta_f\in{[0,2\pi)}$ are the amplitude reflection coefficient and the phase shift of the $f$-th unit cell, respectively. In this paper, we assume $\beta_f=1, \forall f \in \mathcal{F}$ to maximize the signal reflection. Denote $\mathcal{S}_{r}$ as the set of all possible phase shifts for the RIS reflection coefficients, given by
\begin{equation}
\mathcal{S}_{r} \triangleq \left\{e^{j\theta} \; \bigg| \; \theta \in \left\{0,\frac{2\pi}{2^{Q_2}}, \ldots, \frac{2\pi(2^{Q_2}-1)}{2^{Q_2}}\right\} \right\},
\end{equation}
where $Q_2$ is the number of control bits for each RIS element. Again the special case of $Q_2 = \infty$ corresponds to the continuous phase shifts.

We assume that the BS-user direct link is blocked, and thus the direct path can be ignored. {The signal power reflected two or more times is much lower than that reflected just once due to the high free-space path loss. Thus, we ignore the power of the signals that are reflected by the RIS more than once.} In addition, we assume that the channel state information (CSI) of all links involved is perfectly known at the BS and all the channels experience quasi-static flat-fading. How to obtain accurate CSI is an important and challenging issue in the RIS-aided communication system. {The CSI can be obtained through uplink pilots due to the channel reciprocity} and some early attempts can be found in \cite{chen2019channel,9103231, 9127834,9374451,9398559}. Suppose that $\mathbf{G}\in \mathbb{C}^{F \times M}$ is the channel matrix from the BS to the RIS, $\mathbf{h}_{k}^{H} \in\mathbb{C}^{1\times{F}}$ is the channel vector from the RIS to user $k$. Then the received signal of user $k$ can be represented as
\begin{equation}
y_k=\mathbf{h}_{k}^{H}\mathbf{\Theta}\mathbf{G}\mathbf{V}\sum\limits_{j=1}^{K}{\mathbf{w}_js_j}+n_k, \forall{k}\in\mathcal{K},
\end{equation}
where $n_k\thicksim\mathcal{C}\mathcal{N}{(0, \sigma_k^2)}$ is the additive white Gaussian noise at the receiver of user $k$ with zero mean and variance $\sigma^2_k$.  The received SINR of user $k$ can be expressed as
\begin{equation}
\text{SINR}_k=\frac{\arrowvert{\mathbf{h}_{k}^H\mathbf{\Theta}\mathbf{G}\mathbf{V}\mathbf{w}_k} \arrowvert^2}{\sum\limits_{j\not=k}\arrowvert{\mathbf{h}_{k}^H\mathbf{\Theta}\mathbf{G} \mathbf{V}\mathbf{w}_j}\arrowvert^2+\sigma_k^2}, \forall{k}\in\mathcal{K}.
\end{equation}
\subsection{mmWave Channel Model}
We adopt the widely used narrowband clustered channel model \cite{6717211} for mmWave communications. Specifically, the channel matrix between the BS and the RIS can be written as
\begin{equation}
\mathbf{G}={\sqrt{\frac{MF}{{N_\text{cl}}_1{N_\text{ray}}_1}} \sum \limits_{i=1}^{{N_\text{cl}}_1} \sum\limits_{l=1}^{{N_\text{ray}}_1} \alpha_{il}\mathbf{a}_R(\phi_{il}^{Rr}, \delta_{il}^{Rr})\mathbf{a}_B(\phi_{il}^{B}, \delta_{il}^{B})^H}.
\end{equation}
Here, ${N_\text{cl}}_1$ denotes the number of scattering clusters,  ${N_\text{ray}}_1$ denotes the number of rays in each cluster, and $\alpha_{il}$ denotes the channel coefficient of the $l$-th ray in the $i$-th propagation cluster. Moreover, $\mathbf{a}_R(\phi_{il}^{Rr},\delta_{il}^{Rr})$ and $\mathbf{a}_B(\phi_{il}^B,\delta_{il}^B)$ represent the receive array response vectors of the RIS and the transmit array response vectors of the BS respectively, where $\phi_{il}^{Rr}(\phi_{il}^{B})$ and $\delta_{il}^{Rr}(\delta_{il}^B)$ represent azimuth and elevation angles of arrival at the RIS (or departing from the BS).
 The channel vector between the RIS and the $k$-th user can be represented as
\begin{equation}
\mathbf{h}_{k}=\sqrt{\frac{F}{{N_\text{cl}}_2{N_\text{ray}}_2}} \sum \limits_{i=1}^{{N_\text{cl}}_2} \sum\limits_{l=1}^{{N_\text{ray}}_2} \beta_{il}\mathbf{a}_R(\phi_{il}^{Rt},\delta_{il}^{Rt}).
\end{equation}
Here, ${N_\text{cl}}_2$, ${N_\text{ray}}_2$, $\beta_{il}$, $\phi_{il}^{Rt}$ and $\delta_{il}^{Rt}$ are defined in the same way as above.

In this paper, we consider the uniform planar array (UPA) structure at both BS and RIS. Consequently, the array response vector can be denoted as
 \begin{equation}
    \begin{aligned}
      \label{upa}
        \mathbf{a}_{{z}}\left(\phi, \delta\right)= & \frac{1}{\sqrt{A_1 A_2}}\left[1, \ldots, e^{j \frac{2 \pi}{\lambda} d_1 \left(o \sin \phi \sin \delta+p \cos \delta\right)}\right.\\
        & \left.\ldots, e^{j \frac{2 \pi}{\lambda} d_1 \left((A_1-1) \sin \phi \sin \delta)+(A_2-1) \cos \delta\right)}\right]^{T},
    \end{aligned}
  \end{equation}
where ${z} \in \{R,B\}$, $\lambda$ is the signal wavelength, $d$ is the antenna or unit cell spacing which is assumed to be half wavelength distance, $0\leq{o}<A_1$ and $0\leq{p}<A_2$, $A_1$ and $A_2$ represent the number of rows and columns of the UPA in the 2D plane, respectively.
\subsection{Problem Formulation}
\label{pro_formulation}
We consider the QoS problem which aims to minimize the transmit power at the BS by jointly optimizing the digital beamforming matrix  $\mathbf{W}=\left[\mathbf{w}_1, \mathbf{w}_2,\ldots,\mathbf{w}_K \right]\in \mathbb{C}^{N\times K}$ and the analog beamforming matrix $\mathbf{V}$ at the BS, as well as the overall response matrix $\mathbf{\Theta}$ at the RIS, subject to QoS constraints for all users. The problem can be formulated as
    \begin{subequations}\label{prob_original}
    \begin{align}
            {\mathcal{P}_0: \quad \min \limits_{\{\mathbf{V},\mathbf{W},\mathbf{\Theta}\}}} \quad & {D\sum \limits_{k=1}^{K}\left\| \mathbf{w}_{k}\right\|^{2}} \\
            {\text { s.t. }}\quad\quad & {\text{SINR}_{k} \geq \gamma_{k}, \forall k\in\mathcal{K}},  \label{const1}\\
            {}&{{v}_{n,d} \in \mathcal{S}_{a}, \forall n \in \mathcal{N}, \forall d \in \mathcal{D}}, \label{const2}\\
            {}&{ b_f \in \mathcal{S}_{r}, \forall f \in \mathcal{F}}, \label{const3}
    \end{align}
    \end{subequations}


%
where $\gamma_k>0$ is the minimum SINR requirement of user $k$.

The problem $\mathcal{P}_0$ is highly non-convex due to the non-convex SINR constraints \eqref{const1} and the unit-modulus phase shifts constraints \eqref{const2}, \eqref{const3}, and thus difficult to be optimally solved. A commonly used approach to solve such problem approximately is to apply the BCD technique in conjunction with the SDR method as in \cite{wu2019intelligent,9148827}. The BCD technique updates just one block of variables while fixing all the others at a time. In particular, at each iteration, the digital beamforming matrix can be solved via SOCP, while both the analog beamforming matrix and the RIS response matrix can be solved via SDR. Note that SDR cannot guarantee the feasibility due to the rank-one constraint and thus an additional randomization procedure is generally needed. Its complexity is high for the large RIS size. In addition, when the number of users is close to the number of RF chains, the above approach may become invalid because the randomization procedure may fail to find a feasible solution even after a large number of randomization. In this work, we propose a two-layer penalty-based algorithm to solve the problem $\mathcal{P}_0$ as detailed in the next section.

\section{Penalty-based Joint Optimization Algorithm}
\label{sec_qos}
In this section, we propose a two-layer penalty-based method by exploiting the penalty method, where  the BCD method is adopted in the inner layer to solve a penalized problem and the penalty factor is updated in the outer layer until convergence. Specifically, we firstly introduce auxiliary variables $\{t_{k,j}\}$ to represent $ \mathbf{h}_{k}^{H} \mathbf{\Theta} \mathbf{G} \mathbf{V} \mathbf{w}_j$ such that the variables $\mathbf{W}$, $\mathbf{V}$ and $\mathbf{\Theta}$ can be decoupled. Then, the non-convex constraints \eqref{const1} can be equivalently written as
\begin{subequations}
    \begin{align}
        {\frac{ \left| t_{k,k} \right|^{2}} {\sum_{j \neq k}^{K}\left| t_{k,j} \right|^{2}+\sigma_{k}^{2}} \geq \gamma_{k}, \forall k \in \mathcal{K},}\label{penalty_ori_const1}\\
		{t_{k,j}= \mathbf{h}_{k}^{H} \mathbf{\Theta} \mathbf{G} \mathbf{V} \mathbf{w}_j}, \forall k,j \in \mathcal{K}. \label{penalty_ori_const2}
    \end{align}
\end{subequations}

Then, the equality constraints \eqref{penalty_ori_const2} can be relaxed and added to the objective function as a penalty term. Thereby, the original problem $\mathcal{P}_0$ can be converted to the following penalized problem

 \begin{subequations} \label{penalty_prob}
    \begin{align}
            {\mathcal{P}_1(\rho):} {\min \limits_{\small{\mathbf{V}, \mathbf{W}, \mathbf{\Theta},\{t_{k,j}\}}}} \quad & D \sum \limits_{k=1}^{K}\left\| \mathbf{w}_{k}\right\|^{2} \notag \\
            &+ \frac{\rho}{2} \sum \limits_{j=1}^K \sum \limits_{k=1}^K \left| \mathbf{h}_{k}^{H} \mathbf{\Theta} \mathbf{G} \mathbf{V} \mathbf{w}_j-t_{k,j} \right|^2  \\
            {\text { s.t. }}\quad\qquad & {\eqref{penalty_ori_const1},\eqref{const2},\eqref{const3} },
    \end{align}
\end{subequations}
where $\rho>0$ is the penalty factor.
Generally, the choice of $\rho$ is crucial to balance the original objective function and the equality constraints. It is seen that the objective function in $\mathcal{P}_1(\rho)$ is dominated by the penalty term when $\rho$ is large enough and consequently the equality constraints \eqref{penalty_ori_const2} can be well met by the solution. Therefore, we can start with a small value of $\rho$ to get a good starting point, and then by gradually increasing $\rho$, a high precision solution can be obtained.

{There are mainly two different methods to handle the discrete phase shifts. First, the optimal solution can be found by the exhaustive search method \cite{8930608}. However, its complexity is too high to be practical. The second method is to relax the discrete phases to continuous ones and then apply projection \cite{9226616,9133142}. As such, in the rest of the paper, we adopt the projection method. Specifically, we first relax the discrete phase shifts of analog beamforming and RIS coefficients to continuous ones, then solve the relaxed problem with the proposed algorithms, finally project the obtained continuous solution back to the discrete set.}

\subsection{Inner Layer: BCD Algorithm for Solving Problem $\mathcal{P}_1(\rho)$}
\label{sec3a}
For any given $\rho$, though the problem $\mathcal{P}_1(\rho)$ is still non-convex, all the optimization variables $\{\mathbf{W}, \{\mathbf{\Theta}, \mathbf{V}\}, \{t_{k,j}\}\}$ are decoupled in the constraints. We therefore adopt the BCD method to optimize each of them alternately.
\subsubsection{Optimize $\mathbf{W}$}
When $\mathbf{V}$, $\mathbf{\Theta}$ and $\{t_{k,j}\}$ are fixed, problem $\mathcal{P}_1(\rho)$ becomes an unconstrained convex problem. Consequently, the optimal $\mathbf{W}$ can be obtained by the first-order optimality condition, i.e.,
\begin{equation}\label{penalty_optimw}
    \mathbf{w}_k=\rho \mathbf{A}_1^{-1} \sum \limits_{j=1}^K \tilde{\mathbf{h}}_j^H t_{j,k}, \forall k \in \mathcal{K},
\end{equation}
where $ \tilde{\mathbf{h}}_j=\mathbf{h}_{ j}^{H} \mathbf{\Theta} \mathbf{G} \mathbf{V}$ and $\mathbf{A}_1=2D \mathbf{I}_N + \rho \sum \limits_{j=1}^K \tilde{\mathbf{h}}_j^H \tilde{\mathbf{h}}_j$.

\subsubsection{{Optimize $\{\mathbf{\Theta},\mathbf{V}\}$}}
\label{manifold_b}
Let $\mathbf{b}\triangleq[b_1,b_2,\ldots,b_F]^{H}$,
    $
     \mathbf{x} \triangleq\left[\mathbf{v}_{1}^{T}, \mathbf{v}_{2}^{T}, \ldots, \mathbf{v}_{N}^{T}\right]^{T} \in \mathbb{C}^{M\times 1},
    $
    and
    $
    \mathbf{Y}_j \triangleq \text{diag}\{{w}_{ j,1}\mathbf{I}_D, \ldots, {w}_{ j,N}\mathbf{I}_D \} \in \mathbb{C}^{M\times M},
    $
 where $\left|{x}_m\right| = 1, \forall m\in \mathcal{M}\triangleq \{1,2,\ldots,M\}$ and ${w}_{j,n}$ denotes the $n$-th entry of $\mathbf{w}_j$. Then, we can rewrite $\mathbf{V} \mathbf{w}_j=\mathbf{Y}_j \mathbf{x} \in \mathbb{C}^{M\times 1}$ so that the optimization problem is formulated in term of $(\mathbf{b},\mathbf{x})$.
When the digital beamforming matrix $\mathbf{W}$  and the auxiliary variables $\{t_{k,j}\}$  are fixed, the problem $\mathcal{P}_1(\rho)$  is reduced to (with constant terms ignored)
\begin{subequations}\label{optimy1}
\begin{align}
{\min \limits_{\mathbf{b},\mathbf{x} }} &\quad {f(\mathbf{b},\mathbf{x}) =\sum \limits_{j=1}^{K} \sum \limits_{k=1}^{K}\left|\mathbf{b}^{H} \mathbf{c}_{k, j}\mathbf{x}-t_{k, j}\right|^{2}} \\
{\text { s.t. }} &\quad {|b(f)|=1, \forall f \in \mathcal{F}},\label{opty_const1}\\
& \quad{|x(m)|=1, \forall m \in \mathcal{M}},
\end{align}
\end{subequations}
where $\mathbf{c}_{k,j}=\text{diag} (\mathbf{h}_{k}^H) \mathbf{G} \mathbf{Y}_j \in \mathbb{C}^{F\times M}$. In the following, we would like to adopt three different methods to tackle the problem \eqref{optimy1}.

\textbf{\emph{Method One: Alternating Optimization}}
The first idea is to alternately optimize one of the variables $\mathbf{b}$ and $\mathbf{x}$ while keeping the other fixed. When $\mathbf{x}$ is fixed, the main obstacles of the problem \eqref{optimy1} lie in the unit-modulus phase shifts constraints \eqref{opty_const1}. Note that they form a complex circle manifold $\mathcal{M}= \{\mathbf{b}\in \mathbb{C}^F: |b_1|=\cdots=|b_F|=1\}$ \cite{absil2009optimization}. Therefore, the problem \eqref{optimy1} can be efficiently solved by the manifold optimization technique. In specific, we adopt the Riemannian conjugate gradient (RCG) algorithm. The RCG algorithm is widely applied in hybrid beamforming design \cite{yu2016alternating} and recently applied in RIS-aided systems as well \cite{yu2019miso},\cite{guo2020weighted}. In the following we briefly review the general procedure of the RCG algorithm.

Each iteration of the RCG algorithm involves four key steps, namely, to compute the Riemannian gradient, to do the transport, to find the search direction and to do the retraction.

Denote ${f(\mathbf{b}) =\sum \limits_{j=1}^{K} \sum \limits_{k=1}^{K}\left|\mathbf{b}^{H} \mathbf{c}_{k, j}\mathbf{x}-t_{k, j}\right|^{2}}$. For any given point $\mathbf{b}_i$, the Riemannian gradient $\operatorname{grad} f(\mathbf{b}_i)$ is defined as the orthogonal projection of the Euclidean gradient $\nabla f(\mathbf{b}_i)$ onto the tangent space ${T}_{\mathbf{b}_i} \mathcal{M}$ of the manifold $\mathcal{M}$ at point $\mathbf{b}_i$, which can be expressed as
\begin{equation}
T_{\mathbf{b}_i} \mathcal{M}=\left\{\mathbf{b} \in \mathbb{C}^{F}: \operatorname{Re}\left\{\mathbf{b} \odot \mathbf{b}_i^{*}\right\}=\mathbf{0}_{F}\right\}.
\end{equation}
The Euclidean gradient at the point $\mathbf{b}_i$ is given by
\begin{equation}
    \nabla f(\mathbf{b}_i) = 2\sum \limits_{j=1}^K \sum \limits_{k=1}^K \mathbf{c}_{k,j} \mathbf{x} (\mathbf{x}^H \mathbf{c}_{k,j}^H \mathbf{b} - t_{k,j}^H).
\end{equation}
Then, the Riemannian gradient at the point $\mathbf{b}_i$ is given by
\begin{equation} \label{grad}
\operatorname{grad} f(\mathbf{b}_i)=\nabla f(\mathbf{b}_i) -\operatorname{Re}\left\{\nabla f(\mathbf{b}_i) \odot \mathbf{b}_i^*\right\} \odot \mathbf{b}_i.
\end{equation}

With the Riemannian gradient, the optimization technique in the Euclidean space can be extended to the manifold space. Here, we adopt the conjugate gradient method, where the search direction can be updated by
\begin{equation} \label{direction}
\boldsymbol{\eta}_{i+1}=-\operatorname{grad} f(\mathbf{b}_{i+1})+\lambda_1 \mathcal{T}_{\mathbf{b}_{i} \rightarrow \mathbf{b}_{i+1}}\left(\boldsymbol{\eta}_{i}\right),
\end{equation}
where $\boldsymbol{\eta}_i$ is the search direction at $\mathbf{b}_i$,  $\lambda_1$ is the update parameter chosen as the Polak-Ribiere parameter \cite{absil2009optimization}, and $\mathcal{T}_{\mathbf{b}_{i} \rightarrow \mathbf{b}_{i+1}}\left(\boldsymbol{\eta}_{i}\right)$ is the transport operation. Note that $\boldsymbol{\eta}_i$ and $\boldsymbol{\eta}_{i+1}$ lie in different tangent spaces and they cannot be conducted directly. Therefore, the transport operation $\mathcal{T}_{\mathbf{b}_{i} \rightarrow \mathbf{b}_{i+1}}\left(\boldsymbol{\eta}_{i}\right)$ is needed to map the previous search direction from its original tangent space to the current tangent space at the current point $\mathbf{b}_{i+1}$. The transport operation is given by
\begin{equation} \label{transport}
\begin{aligned}
\mathcal{T}_{\mathbf{b}_{i} \rightarrow \mathbf{b}_{i+1}}\left(\boldsymbol{\eta}_{i}\right) : T_{\mathbf{b}_{i}} \mathcal{M} & \mapsto T_{\mathbf{b}_{i+1}} \mathcal{M}: \\
\boldsymbol{\eta}_{i} & \mapsto \boldsymbol{\eta}_{i}-\operatorname{Re}\left\{\boldsymbol{\eta}_{i} \odot \mathbf{b}_{i+1}^{*}\right\} \odot \mathbf{b}_{i+1}.
\end{aligned}
\end{equation}

Since the updated point may leave the previous manifold space, a retraction operation $\operatorname{Retr}_{\mathbf{b}}(\lambda_2 \boldsymbol{\eta}_i)$ is needed to project the point back to the manifold:
\begin{equation} \label{retraction}
\begin{aligned}
\operatorname{Retr}_{\mathbf{b}_i}(\lambda_2 \boldsymbol{\eta}_i) : T_{\mathbf{b}_{i}} \mathcal{M} & \mapsto \mathcal{M}: \\
\lambda_2 \boldsymbol{\eta}_{i} & \mapsto \frac{\left(\mathbf{b}_i+\lambda_{2} \boldsymbol{\eta}_{i}\right)_{j}}{\left|\left(\mathbf{b}_i+\lambda_{2} \boldsymbol{\eta}_{i}\right)_{j}\right|},
\end{aligned}
\end{equation}
where $\lambda_2$ is the Armijo backtracking line search step size, and $(\mathbf{b}_i+\lambda_{2} \boldsymbol{\eta}_{i})_j$ denotes the $j$-th entry of $\mathbf{b}_i+\lambda_{2} \boldsymbol{\eta}_{i}$.

The key steps are introduced above, and the consequent algorithm for solving the problem \eqref{optimy1} with fixed $\mathbf{x}$ is summarized in \textit{Algorithm} \ref{rcg}. \textit{Algorithm} \ref{rcg} is guaranteed to converge to a stationary point \cite{absil2009optimization}.

When $\mathbf{b}$ is fixed, $\mathbf{x}$ can be also updated similarly by the RCG algorithm.

 \begin{algorithm}[t]
    \caption{RCG Algorithm for solving problem \eqref{optimy1} with fixed $\mathbf{x}$}
    \label{rcg}
    \hspace*{0.02in} {\bf Input:} $\{\mathbf{c}_{k,j}\}$, $\mathbf{x}$,  $\mathbf{b}_0 \in \mathcal{M}$
    \begin{algorithmic}[1]
    \State Calculate $\boldsymbol{\eta}_0= -\operatorname{grad} f(\mathbf{b}_0)$ according to \eqref{grad} and set $i=0$;
    \Repeat
            \State Choose the Armijo backtracking line search step size $\lambda_2$;
            \State  Find the next point $\mathbf{b}_{i+1}$ using retraction according to \eqref{retraction};
            \State Calculate the Riemannian gradient $\operatorname{grad} f(\mathbf{b}_{i+1})$ according to \eqref{grad};
            \State Calculate the transport $\mathcal{T}_{\mathbf{b}_{i} \rightarrow \mathbf{b}_{i+1}} \left(\boldsymbol{\eta}_{i}\right)$ according to \eqref{transport};
            \State Choose the Polak-Ribiere parameter $\lambda_1$;
            \State Calculate the conjugate direction $\boldsymbol{\eta}_{i+1}$ according to \eqref{direction};
            \State $i \leftarrow i+1$;
    \Until $\|\operatorname{grad} f(\mathbf{b}_i)\|_2 \leq \epsilon_1$.
    \end{algorithmic}
 \end{algorithm}

\textbf{\emph{Method Two: RCG-based Joint Optimization}}
Note that both $\mathbf{b}$ and $\mathbf{x}$ of the problem \eqref{optimy1} are subject to unit-modulus constraints. Thus we can concatenate them and treat as a higher-dimensional vector subject to the same unit-modulus constraints. Specifically, let $\mathbf{z}=\left[\mathbf{b}^H, \mathbf{x}^H\right]^H \in \mathbb{C}^ {(F+M)\times 1}$, and we can rewrite the problem \eqref{optimy1} as follows
\begin{subequations} \label{penalty_optimb}
    \begin{align}
        {\min \limits_{\mathbf{z}}} \quad & { f(\mathbf{z})=\sum \limits_{j=1}^{K} \sum \limits_{k=1}^{K}\left| \mathbf{z}^H \mathbf{d}_{k,j} \mathbf{z}-t_{k,j} \right|^2} \label{penalty_obtimb_obj}\\
        {\text { s.t. }} \quad & { |z(i)|=1, \forall i \in \mathcal{Z}}\label{unitm},
    \end{align}
\end{subequations}
where $\mathbf{d}_{k,j}=\left[\begin{array}{l}
\mathbf{I}_{F \times F} \\
\mathbf{O}_{M \times F}
\end{array}\right] \mathbf{c}_{k, j}\left[\mathbf{O}_{M \times F}\quad \mathbf{I}_{M \times M}\right] \in \mathbb{C}^{(M+F)\times (M+F)}$ and $\mathcal{Z}\triangleq \{1,2, \ldots, F+M\}$.
{The Euclidean gradient of the function $f(\mathbf{z})$ is given by}
\begin{equation}
\nabla f(\mathbf{z})=\left[\begin{array}{c}
2 \sum \limits_{j=1}^{K} \sum_{k=1}^{K} \mathbf{c}_{k, j} \mathbf{x}\left(\mathbf{x}^{H} \mathbf{c}_{k, j}^{H} \mathbf{b}-t_{k, j}^{H}\right) \\
2 \sum \limits_{j=1}^{K} \sum_{k=1}^{K} \mathbf{c}_{k,j}^{H} \mathbf{b}\left(\mathbf{b}^{H} \mathbf{c}_{k, j} \mathbf{x}-t_{k,j}\right)
\end{array}\right].
\end{equation}
Therefore, the problem \eqref{penalty_optimb} can be effectively solved by the RCG algorithm.

Note that the objective function of the problem \eqref{optimy1} is convex over $\mathbf{b}$ or $\mathbf{x}$. In the alternating optimization, the subproblem is reduced to an unconstrained convex problem in the manifold space. Therefore, the optimal solution can be obtained for each subproblem by the RCG algorithm. However, the function $f(\mathbf{z})$ is not jointly convex in $\mathbf{b}$ and $\mathbf{x}$. Thus, in the RCG-based joint optimization, only the sub-optimal solution can be obtained.

\textbf{\emph{Method Three: SCA-based Joint Optimization}}
The RCG algorithm requires multiple projections. If we directly optimize the phase shifts, the projection procedure is no longer needed. Then the problem \eqref{penalty_optimb} becomes an unconstraint non-convex problem, i.e.,
\begin{equation} \label{ques2}
\begin{array}{cc}
{\min \limits_{\bs\phi}} & {f(\bs\phi)=\sum \limits_{j=1}^K \sum \limits_{k=1}^K \left| (e^{j\bs\phi})^H\mathbf{d}_{k,j} e^{j\bs\phi}
- t_{k,j} \right|^2},
\end{array}
\end{equation}
where $\bs{\phi}=\angle{\mathbf{z}}$. Though the above problem is still difficult to solve optimally, we only need to solve its surrogate problem by exploiting the SCA technique, and the BCD method will converge to a stationary solution \cite{razaviyayn2013unified}. Specifically, denote the surrogate function for $f(\bs\phi)$ by $g(\bs \phi, \bar{\bs\phi})$. Then, $\bs\phi$ can be updated by solving the following surrogate problem
\begin{equation}
\bs\phi=\arg \min \limits_{\bs\phi \in \mathbb{R}^{F+M}} g(\bs \phi, \bar{\bs\phi}).
\end{equation}
The surrogate function $g(\bs \phi, \bar{\bs\phi})$ needs to satisfy following the two constraints \cite[Proposition 1]{razaviyayn2013unified}:
\begin{subequations}
\begin{align}
g(\bar{\bs\phi}, \bar{\bs\phi})=f(\bar{\bs \phi}), \\
g(\bs \phi, \bar{\bs \phi}) \geq f (\bs \phi). \label{surrogate2}
\end{align}
\end{subequations}
We can construct the surrogate function by the second order Taylor expansion:
\begin{equation}
g(\bs\phi, \bar{\bs\phi})=f (\bar{\bs\phi})+\nabla f (\bar{\bs\phi})^{T} (\bs\phi-\bar{\bs\phi})+\frac{1}{2\kappa}\|\bs\phi-\bar{\bs\phi}\|^{2},
\end{equation}
where $\nabla f (\bar{\bs\phi})$ is the gradient, and $\kappa$ is chosen to satisfy \eqref{surrogate2} locally within a bounded feasible set. Then, $\bs \phi$ is updated by
\begin{equation}
\bs \phi = \bar{\bs\phi}- \kappa \nabla f(\bar{\bs\phi}).
\end{equation}
In practice, the parameter $\kappa$ can be determined by the Armijo rule:
\begin{equation}\label{bu}
f(\bar{\bs\phi})- f(\bs\phi) \geq \zeta\kappa \|\nabla f(\bar{ \bs\phi}) \|^2,
\end{equation}
where $0<\zeta<0.5$, $\kappa$ is the largest element in $\{\beta\kappa_0^i\}_{i=0,1,\ldots}$ that makes \eqref{bu} satisfied,  $\beta>0$ and $0<\kappa_0<1$.

\subsubsection{Optimize $\{t_{k,j}\}$}
With other variables fixed, problem $\mathcal{P}_1(\rho)$ can be reduced to
\begin{subequations} \label{penalty_t}
    \begin{align}
        {\min \limits_{\{t_{k,j}\}}} \quad & { \sum \limits_{j=1}^K \sum \limits_{k=1}^K \left| \mathbf{h}_{k}^{H} \mathbf{\Theta} \mathbf{G} \mathbf{V} \mathbf{w}_j-t_{k,j} \right|^2} \\
        {\text { s.t. }} \quad & {\frac{ \left| t_{k,k} \right|^{2}} {\sum_{j \neq k}^{K}\left| t_{k,j} \right|^{2}+\sigma_{k}^{2}} \geq \gamma_{k}, \forall k \in \mathcal{K}}\label{22}.
    \end{align}
\end{subequations}
The objective function is convex over $\{t_{k,j}\}$. Although the constraints \eqref{22} are non-convex, they can be translated to the form of second-order cones as follows,
\begin{equation}
\sqrt{1+\frac{1}{\gamma_k}}t_{k,j} \geq \left\|\begin{array}{c}
\mathbf{A}_2^{H} \mathbf{e}_{k} \\
\sigma_{k}
\end{array}\right\|_{2}, \forall k \in \mathcal{K},
\end{equation}
where $\mathbf{A}_2\in \mathbb{C}^{K \times K}$ denotes a matrix with the entry in its $k$-th row and $j$-the column being $t_{k,j}$, i.e., $\mathbf{A}_2[k,j]=t_{k,j}$, and $\mathbf{e}_k\in \mathbb{C}^{K\times 1}$ denotes a vector with the $k$-th entry being one and others being zeros. Then, the problem \eqref{penalty_t} can be effectively and optimally solved by the SOCP method \cite{1561584}.

\subsection{Outer Layer: Update Penalty Factor}
The penalty factor $\rho$ is initialized to be a small number to find a good starting point, then gradually increased to tighten the penalty. Specifically,
\begin{equation}\label{penalty_rho}
    \rho:=\frac{\rho}{c}, 0 <c <1,
\end{equation}
where $c$ is a constant scaling parameter. A larger $c$ may lead to a more precise solution with a longer running time.

\subsection{Algorithm}
 \begin{algorithm}[t]
    \caption{Penalty-based Optimization Algorithm}
    \label{alg_penalty}
    \begin{algorithmic}[1]
    \State Initialize $\mathbf{V}$, $\mathbf{\Theta}$, $\rho$ and $\{t_{k,j}\},\forall k,j \in \mathcal{K}$.
    \Repeat
        \Repeat
            \State Update $\mathbf{W}$ by \eqref{penalty_optimw};
            \State Update $\mathbf{\Theta}$ and $\mathbf{V}$ by solving problem \eqref{optimy1};
            \State Update $\{t_{k,j}\}$ by solving problem \eqref{penalty_t};
        \Until The decrease of the objective value of problem $\mathcal{P}_1(\rho)$ is below threshold $\epsilon_2>0$.
        \State Update $\rho$ by \eqref{penalty_rho}.
    \Until The stopping indicator $\xi$ in \eqref{stop_criteria} is below threshold $\epsilon_3>0$.
    \State Project $\mathbf{\Theta}$ and $\mathbf{V}$ onto the discrete sets $\mathcal{S}_{r}$ and $\mathcal{S}_{a}$, respectively;
    \State Update $\mathbf{W}$ by solving problem \eqref{inde_w} with the projected $\mathbf{\Theta}$ and $\mathbf{V}$.
    \end{algorithmic}
 \end{algorithm}

The overall penalty-based optimization algorithm is summarized in \textit{Algorithm} \ref{alg_penalty}. Define the stopping indicator $\xi$ as follows,
\begin{equation}\label{stop_criteria}
\xi \triangleq \max \left\{ | \mathbf{h}_{k}^{H} \mathbf{\Theta} \mathbf{G} \mathbf{V} \mathbf{w}_j-t_{k,j} |^2, \forall k,j \in \mathcal{K} \right\}.
\end{equation}
When $\xi$ is below a pre-defined threshold $\epsilon_3>0$, the equality constraints \eqref{penalty_ori_const2} are considered to be satisfied and the proposed algorithm is terminated. Since we start with a small penalty and gradually increase its value, the objective value of problem $\mathcal{P}_1(\rho)$ is finally determined by the penalty part and the equality constraints are guaranteed to be satisfied. {Note that, for any given penalty factor $\rho$, the objective value of the problem $\mathcal{P}_1(\rho)$ solved through the BCD method is non-increasing over iterations in the inner layer. And the optimal value of the problem $\mathcal{P}_1(\rho)$ is bounded by the SINR constraints. Thereby, based on the \textit{Theorem 4.1} of the work \cite{7558213}, the proposed \textit{Algorithm} \ref{alg_penalty} is guaranteed to converge.}

{Let us consider the complexity of the proposed algorithm. Let us first compare the complexities of the three different methods, which are dominated by computing the Euclidean gradient. Thus, the complexity of Alternating Opt is $\mathcal{O}(I_\mathbf{b} K^2 F+ I_\mathbf{x} K^2 M)$, where $I_\mathbf{b}$ and $I_\mathbf{x}$ denote the required iteration times of the RCG algorithm to update $\mathbf{b}$ and $\mathbf{x}$, respectively. The complexity of RCG-based Joint Opt is $\mathcal{O}(I_\mathbf{z} K^2 (F+M))$, where $I_\mathbf{z}$ denotes the required iteration times of the RCG algorithm to update $\mathbf{z}$.  The complexity of SCA-based Joint Opt is $\mathcal{O}(I_a K^2 (F+M))$, where $I_a$ denotes the iteration number of the Armijo search. As will be shown in Section \ref{simulation_methods}, the RCG-based joint optimization method outperforms the other two methods. Thus, we adopt the RCG-based joint optimization method and analyze its complexity. It can be shown that the complexity of computing $\mathbf{W}$ in \eqref{penalty_optimw} is $\mathcal{O}(N^3+KN^2 +K^2N)$. Besides, the complexity of solving problem \eqref{penalty_t} is $\mathcal{O}(K^{3.5})$. Thereby, the overall complexity of \textit{Algorithm} \ref{alg_penalty} is $\mathcal{O}(I_{out}I_{in}(N^3+KN^2 +K^2N +I_{\mathbf{z}}K^2 (F+M) +K^{3.5}))$ where $I_{out}$ and $I_{in}$ denote the outer and inner iteration times required for convergence, respectively.}


\section{Sequential Optimization}
\label{individual}
{To reduce the complexity of solving the problem $\mathcal{P}_0$, we develop a sequential optimization approach in this section. Specifically, we first optimize the RIS response matrix $\mathbf{\Theta}$, then optimize the analog beamformer $\mathbf{V}$, and finally optimize the digital beamformer $\mathbf{W}$ without iteration.}

\subsection{RIS Design}
\label{ris_design}
Looking at the SINR constraints \eqref{const1}, and we can get
\begin{equation}\label{ris1}
|\mathbf{h}_k^H \mathbf{\Theta} \mathbf{G}\mathbf{V}\mathbf{w}_k| - \gamma_k \sum _{j\not=k} |\mathbf{h}_{k}^H \mathbf{\Theta} \mathbf{G} \mathbf{V}\mathbf{w}_j| \geq 0, \forall k \in \mathcal{K}.
\end{equation}
For simplicity, let the transmit beamforming vectors at the BS be set based on the maximum-ratio transmission (MRT) principle, i.e., $\mathbf{V}\mathbf{w}_k= (\mathbf{h}_k^H \mathbf{\Theta} \mathbf{G})^H$. {Note that the transmit beamforming vectors here are only used to extract the optimization of the RIS response matrix. The actually adopted transmit beamforming vectors are designed later.} Then, the problem \eqref{ris1} is translated to
\begin{equation}\label{ris2}
\|\mathbf{h}_k^H \mathbf{\Theta} \mathbf{G}\|^2 - \gamma_k \sum _{j\not=k} \|\mathbf{h}_{k}^H \mathbf{\Theta} \mathbf{G} \mathbf{G}^H \mathbf{\Theta}^H \mathbf{h}_j\| \geq 0, \forall k \in \mathcal{K}.
\end{equation}
{The inequality \eqref{ris2} should be satisfied for all users. Therefore, in order to ensure the receive signal quality of each user, we maximize the worst case of the left-hand side of \eqref{ris2} among all users, i.e.,}
\begin{subequations}\label{max_min_ris}
        \begin{align}
            {\max \limits_{\mathbf{\Theta}}}\quad &{\min \limits_{k\in \mathcal{K}} \ \|\mathbf{h}_k^H \mathbf{\Theta} \mathbf{G}\|^2 - \gamma_k \sum _{j\not=k} \|\mathbf{h}_{k}^H \mathbf{\Theta} \mathbf{G} \mathbf{G}^H \mathbf{\Theta}^H \mathbf{h}_j\|} \\
            {\text { s.t. }}\quad & {b_f \in \mathcal{S}_{r}, \forall f \in \mathcal{F}}.
        \end{align}
\end{subequations}
The SDR technique can be adopted to solve the above problem. Specifically, let us introduce an auxiliary variable $\varpi$, and let $\mathbf{B}=\mathbf{b}\mathbf{b}^H$. After dropping the rank-one constraint, the problem \eqref{max_min_ris} can be relaxed into
\begin{subequations}\label{max_min_ris2}
        \begin{align}
            {\max \limits_{\mathbf{B}, \varpi}}\quad &{\varpi} \\
            {\text { s.t. }} \quad& { \text{Tr}(\mathbf{\eta}_k \eta_k^H \mathbf{B}) \geq \varpi + \gamma_k \sum \limits_{j\not=k} \|\text{Tr}(\zeta_{k,j} \mathbf{B})\|, \forall k \in \mathcal{K}, }\\
            {}&{\mathbf{B}_{f,f}=1, \forall f \in \mathcal{F}},\\
            {}&{\mathbf{B} \succeq 0, }
        \end{align}
\end{subequations}
where $\eta_k=\text{diag}(\mathbf{h}_k^H)\mathbf{G}\in \mathbb{C}^{F\times M}$ and $\zeta_{k,j}=\text{diag}(\mathbf{h}_k^H) \mathbf{G} \mathbf{G}^H \text{diag}(\mathbf{h}_j) \in \mathbb{C}^{F\times F}$. The problem \eqref{max_min_ris2} is convex and can be optimally solved by a standard convex solver such as CVX\cite{grant2014cvx}. After solving the problem \eqref{max_min_ris2}, the optimal $\mathbf{B}$ can be obtained. Then, we need to obtain the value of $\mathbf{b}$, which has the direct relationship to $\mathbf{B}$. Generally, there is no guarantee that the relaxed problem \eqref{max_min_ris2} has a rank-one optimal solution. If $\text{rank}(\mathbf{B}) = 1$, then we can obtain the optimal $\mathbf{b}$ by taking the eigenvalue decomposition of $\mathbf{B}$. Otherwise, if $\text{rank}(\mathbf{B}) > 1$, an additional Gaussian randomization procedure is needed to produce a rank-one solution \cite{5447068,wu2019intelligent}. Specifically, suppose that the eigenvalue decomposition of $\mathbf{B}$ is $\mathbf{B} = \mathbf{U}\Sigma\mathbf{U}^H$. Then, let $\overline{\mathbf{b}} = \mathbf{U}\Sigma^{1/2}\mathbf{r}$, where $\mathbf{r} \sim \mathcal{C} \mathcal{N}\left(0, \boldsymbol{I}_{F}\right)$. Finally, project $\overline{\mathbf{b}}$ to the pre-defined set $\mathcal{S}_r$, i.e.,
\begin{equation}
    b_f=e^{j\angle b_f},
\end{equation}
where $\angle b_f =\arg \min \limits_{\angle b_f \in \mathcal{S}_{r}} |\angle b_f  -\angle \overline{b}_f|$. With many independently generated $\mathbf{r}$, the one that makes $\varpi$ maximum is taken as the solution.

\subsection{Analog Beamforming Design}
{We then optimize the analog beamforming after the RIS has been configured.} The orthogonal match pursuit (OMP) method is widely adopted to design the analog beamformer\cite{6717211}. If the BS adopts the fully digital beamforming structure, the optimal digital beamforming $\mathbf{W}_{\text {opt}}$ can be obtained by solving the following problem
\begin{subequations} \label{ind_wopt}
        \begin{align}
            {\min \limits_{\mathbf{W}}}\quad &{D\sum \limits_{k=1}^{K}\left\| \mathbf{w}_{k}\right\|^{2}} \\
            {\text { s.t. }} \quad&  { \frac{ |\mathbf{h}_{k}^H \mathbf{\Theta} \mathbf{G} \mathbf{w}_k|^2}{\sum \limits_{j\not= k} |\mathbf{h}_{k}^H \mathbf{\Theta} \mathbf{G} \mathbf{w}_j|^2 + \sigma^2_k} \geq \gamma_k, \forall k \in \mathcal{K}}.
        \end{align}
\end{subequations}
{Note} that the above problem can be optimally solved by the SOCP method. {We adopt an overlapping codebook $\mathbf{A}$ with an overlapping coefficient $\mu$ to improve the spatial resolution due to the limited resolution of the conventional DFT codebook \cite{8362957}. A larger $\mu$ represents higher resolution of the codebook.} The codebook can be represented as $\mathbf{A}=[\mathbf{a}_B (\psi_1, \phi_1), \ldots, \mathbf{a}_B (\psi_1, \phi_{\mu N_z}), \ldots, \newline \mathbf{a}_B (\psi_{\mu N_y}, \phi_{\mu N_z})]$, where $N_y$ and $N_z$ denote the horizontal and vertical lengths, $\psi_i= \frac{2 \pi(i-1)}{\mu N_{y}}, i=1,2,\ldots,\mu N_y$ and
$\phi_j =\frac{2 \pi(j-1)}{\mu N_{z}}, j=1,2,\ldots,\mu N_z$, respectively.
Then, we can use a selection matrix $\mathbf{T}\in \mathbb{R}^{\mu^2 N_y N_z \times N}$ to select proper columns. Specifically, the analog beamforming problem can be formulated as
\begin{subequations} \label{analog_codebook}
    \begin{align}
        {\mathbf{T}^*=\underset{\mathbf{T}, \mathbf{F}_{BB}}{\arg \min } } \quad & {  \left\|\mathbf{W}_{\text {opt}}- \mathbf{A}_t \mathbf{T} \mathbf{F}_{BB} \right\|_{F} }\\
         {\text{ s.t. }}\quad & {\left\|\operatorname{diag}\left(\mathbf{T} \mathbf{T}^{H}\right)\right\|_{0}= N, }
    \end{align}
\end{subequations}
{where $\mathbf{A}_t= \mathbf{e}_{t} \odot \mathbf{A}, t\in \mathcal{N}$, and $\mathbf{e}_t$ is an $M\times1$ zero-vector with the entry from $(t-1)D+1$ to $tD$ being one.} Since the structure of analog beamforming is sub-connected, we use $\mathbf{e}_t$ to modify the codebook. Then, the OMP method can be applied to obtain the selection matrix $\mathbf{T}^*$. The analog beamforming can be recovered, i.e., $\mathbf{V}=\mathbf{A}_t \mathbf{T}^*$.  Finally, the discrete analog beamforming can be obtained by mapping $\mathbf{V}$ to the nearest discrete value in $\mathcal{S}_{a}$.

\subsection{Digital Beamforming Design}
After obtaining the RIS phase shifts and the analog beamforming vector, we need to obtain the optimal digital beamforming matrix. The digital beamforming can be obtained by solving following problem
\begin{subequations} \label{inde_w}
        \begin{align}
            {\min \limits_{\mathbf{W}}}\quad &{D\sum \limits_{k=1}^{K}\left\| \mathbf{w}_{k}\right\|^{2}} \\
            {\text { s.t. }} \quad&  { \frac{ |\mathbf{h}_{k}^H \mathbf{\Theta} \mathbf{G} \mathbf{V} \mathbf{w}_k|^2}{\sum \limits_{j\not= k} |\mathbf{h}_{k}^H \mathbf{\Theta} \mathbf{G} \mathbf{V} \mathbf{w}_j|^2 + \sigma^2_k} \geq \gamma_k, \forall k \in \mathcal{K}}.
        \end{align}
\end{subequations}
{Note that the digital beamforming $\mathbf{W}_{\text {opt}}$ obtained by solving the problem \eqref{ind_wopt} is only used for the analog beamforming design. }The problem \eqref{inde_w} is the conventional power minimization problem in the multiple-input-single-output (MISO) system, which can be effectively and optimally solved by the SOCP method \cite{1561584}.

{Here, we consider the complexity of the sequential optimization. The complexity of the RIS design is dominated by the SDR technique, which is $\mathcal{O}(F^6)$ \cite{ben2001lectures}. The complexity of the analog beamforming is dominated by the OMP technique, which is $\mathcal{O}(\mu^2 MFN^3)$. The complexity of the digital beamforming design is $\mathcal{O}(N^{3.5} K^{3.5})$ \cite{7961152}. Thus, the overall computational complexity of the Sequential Optimization is $\mathcal{O}(F^6+\mu^2 MFN^3+N^{3.5} K^{3.5})$. The advantage of this algorithm is that it does not need to perform iterative operations.}

\section{Extension to the Max-Min Fairness Problem}
\label{discussion}
A closely related problem of the QoS problem $\mathcal{P}_0$ is the MMF problem, which aims to maximize the performance of the worse-case user under a fixed total transmit power budget. In this section, we discuss the relationship between the QoS problem and the MMF problem, and the extension of the proposed algorithm to solve the MMF problem. In specific, the MMF problem is to maximize the weighted minimum SINR under a total power budget $P_T$, which can be formulated as
\begin{subequations}\label{prob_original2}
    \begin{align}
            {\mathcal{Q}_0: \max \limits_{\{\mathbf{V},\mathbf{W},\mathbf{\Theta}\}} \min \limits_{k\in \mathcal{K}}} \quad & { \frac{1}{\gamma_k}\frac{\arrowvert{\mathbf{h}_{k}^H\mathbf{\Theta}\mathbf{G}\mathbf{V}\mathbf{w}_k} \arrowvert^2}{\sum\limits_{j\not=k}\arrowvert{\mathbf{h}_{k}^H\mathbf{\Theta}\mathbf{G} \mathbf{V}\mathbf{w}_j}\arrowvert^2+\sigma_k^2} } \\
            {\text { s.t. }}\quad\quad & {D\sum \limits_{k=1}^K \|\mathbf{w}_k\|^2 \leq P_T }, \label{prob_f_cons1}\\
            {}&{\eqref{const2},\eqref{const3}}
    \end{align}
\end{subequations}
where $\gamma_k>0$ denotes the weight parameter of user $k$. A larger value of $\gamma_k$ indicates that user $k$ has a higher priority in transmission.


Let us compare the problem $\mathcal{P}_0$ and  the problem $\mathcal{Q}_0$. Let $\bs{\gamma}\triangleq[\gamma_1, \gamma_2,\ldots, \gamma_K]^T$. For a given set of channels and noise powers, $\mathcal{P}_0$ is parameterized by $\bs{\gamma}$. We use the notation $\mathcal{P}_0 (\bs{\gamma})$ to account for this, and $P_T=\mathcal{P}_0 (\bs{\gamma})$ to denote the associated minimum power. Similarly, $\mathcal{Q}_0$ is parameterized by $\bs{\gamma}$ and $P_T$. Then, $\mathcal{Q}_0 (\bs{\gamma}, P_T)$ and $\xi=\mathcal{Q}_0 (\bs{\gamma}, P_T)$ are used to represent the dependence and the associated maximum worst-case weighted SINR, respectively. Similar to \cite{1634819,4443878}, we have the following proposition.
\begin{proposition}\label{remark}
The QoS problem $\mathcal{P}_0$ and the MMF problem $\mathcal{Q}_0$ have the following relationship:
\begin{subequations}
    \begin{align}
         \xi=\mathcal{Q}_0(\bs{\gamma},\mathcal{P}_0(\xi\bs{\gamma})),\label{relation1}\\
         P_T=\mathcal{P}_0(\mathcal{Q}_0(\bs{\gamma},P_T)\bs{\gamma}).\label{relation2}
    \end{align}
\end{subequations}
\end{proposition}
\begin{proof}
Contradiction argument is used to prove \eqref{relation1}. For the problem $\mathcal{P}_0(\xi\bs{\gamma})$, denote the optimal solution and the associated optimal value as $\{\mathbf{W}^{\mathcal{P}_0}, \mathbf{\Theta}^{\mathcal{P}_0}, \mathbf{V}^{\mathcal{P}_0} \}$ and $P_T^{\mathcal{P}_0}$, respectively. It is observed that the set $\{\mathbf{W}^{\mathcal{P}_0}, \mathbf{\Theta}^{\mathcal{P}_0}, \mathbf{V}^{\mathcal{P}_0} \}$ is also a feasible solution with the objective value $\xi$ to the problem $\mathcal{Q}_0(\bs{\gamma}, P_T^{\mathcal{P}_0})$. Since $\mathbf{\Theta}$ and $\mathbf{V}$ have unit-modulus constraints, we can only scale $\mathbf{W}$. Assume there is another solution $\{\mathbf{W}^{\mathcal{Q}_0}, \mathbf{\Theta}^{\mathcal{P}_0}, \mathbf{V}^{\mathcal{P}_0} \}$ with bigger objective value $\xi^{\mathcal{Q}_0}>\xi$. Then, we can appropriately scale down the digital beamforming with the SINR constraints of the problem $\mathcal{P}_0(\xi\bs{\gamma})$ still satisfied. The resulting solution $\{c\mathbf{W}^{\mathcal{Q}_0}, \mathbf{\Theta}^{\mathcal{P}_0}, \mathbf{V}^{\mathcal{P}_0} \}(0<c<1)$ has a smaller transmit power than $P_T^{\mathcal{P}_0}$, which contradicts the optimality of $\{\mathbf{W}^{\mathcal{P}_0}, \mathbf{\Theta}^{\mathcal{P}_0}, \mathbf{V}^{\mathcal{P}_0} \}$. \eqref{relation2} can be proved in the similar way and the details are omitted.
\end{proof}

Generally, the MMF problem $\mathcal{Q}_0$ is more difficult to solve than the QoS problem $\mathcal{P}_0$ due to {the non-smooth} objective function. Based on Proposition \ref{remark}, we can solve the MMF problem by solving a series of QoS problems. Specifically, let us consider the following problem $\mathcal{P}_2 (\varsigma)$, i.e.,
    \begin{subequations}
    \begin{align}
            {\mathcal{P}_2 (\varsigma): \min \limits_{\{\mathbf{V},\mathbf{W},\mathbf{\Theta}\}}} \quad & {D\sum \limits_{k=1}^{K}\left\| \mathbf{w}_{k}\right\|^{2}} \\
            {\text { s.t. }}\quad\quad & {\text{SINR}_{k} \geq \varsigma \gamma_{k}, \forall k\in\mathcal{K}}, \\
            {}&{\eqref{const2}, \eqref{const3}}.
    \end{align}
    \end{subequations}
For a given set of channels, noise powers and $\bs{\gamma}$, $\mathcal{P}_2$ is parameterized by $\varsigma$. Note that the problem $\mathcal{P}_2 (\varsigma)$ is a linear function over $\varsigma$. A larger $\varsigma$ leads to a larger objective value of $\mathcal{P}_2$. Thus, in order to solve the problem $\mathcal{Q}_0$, we can do a bisection search over $\varsigma$ of the problem $\mathcal{P}_2 $ until its objective value is $P_T$. Then, the corresponding result is the solution to $\mathcal{Q}_0$ with the total power budget being $P_T$.

\section{Simulation Results}
\label{sec_simulation}
\begin{figure}[t]
\begin{centering}
\vspace{-0.5cm}
\includegraphics[width=0.45\textwidth]{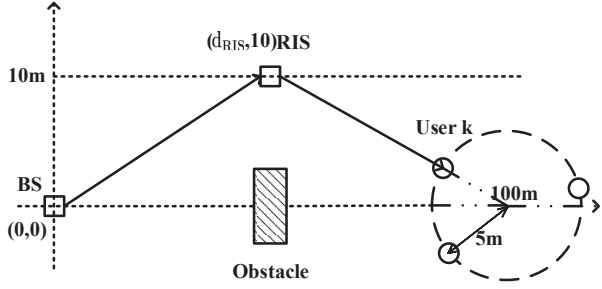}
 \caption{The simulated RIS-aided communication scenario.}\label{simulation_setup}
\end{centering}
\vspace{-0.3cm}
\end{figure}
In this section, we evaluate the performance of our proposed algorithms. {We consider an RIS-aided multiuser mmWave communication system which operates at 28 GHz with bandwidth 251.1886 MHz. Thus, the noise power is $\sigma_k^2=-174+10 \log _{10} B=-90$ dBm.} We consider a $6\times 6$ UPA structure at the BS with $N=6$ RF chains and a total of $M=36$ antennas located at (0 m, 0 m) as shown in Fig.~\ref{simulation_setup}. The RIS is located at ($d_{RIS}$ m, 10 m) and equipped with $F_1 \times F_2$ unit cells where $F_1=6$ and $F_2$ can vary. Users are uniformly and randomly distributed in a circle centered at (100 m, 0 m) with radius 5 m. As for the mmWave channel, we set ${N_{\text{cl}}}_1={N_\text{cl}}_2=5$ clusters, ${N_\text{ray}}_1={N_\text{ray}}_2=10$ rays per cluster; the azimuth and elevation angles of arrival and departure follow the Lapacian distribution with an angle spread of 10 degrees; the complex gain $\alpha_{il}$ and $\beta_{il}$ follow the complex Gaussian distribution $\mathcal{CN}(0,10^{-0.1PL(d)})$, and $PL(d)$ can be modeled as\cite{akdeniz2014millimeter}:
\begin{equation}
PL(d)=\varphi_a + 10 \varphi_b \log_{10} (d) + \varphi_c (\text{dB}),
\end{equation}
where $\varphi_c \sim \mathcal{N}\left(0, \sigma^{2}\right)$, $\varphi_a=72.0,\varphi_b=2.92$ and $\sigma=8.7$dB. The auxiliary variables $\{t_{k,j}\}$ are initialized following $\mathcal{C}\mathcal{N}(0,1)$. The penalty factor is initialized as $\rho=10^{-3}$. Other system parameters are set as follows unless specified otherwise later: $K=3, F_2=6, d_{RIS}=50, c=0.9, \epsilon_1=\epsilon_3=10^{-7}, \epsilon_2=10^{-4}, \gamma_k=10$dB,  $\forall k \in \mathcal{K}$. All simulation curves are averaged over $100$ independent channel realizations. {The simulations are carried out on a computer with Intel i7-7700 CPU at 3.60 GHz and with 16.0 GB RAM.}

\begin{figure}[t]
\begin{centering}
\vspace{-0.1cm}
\includegraphics[width=0.45\textwidth]{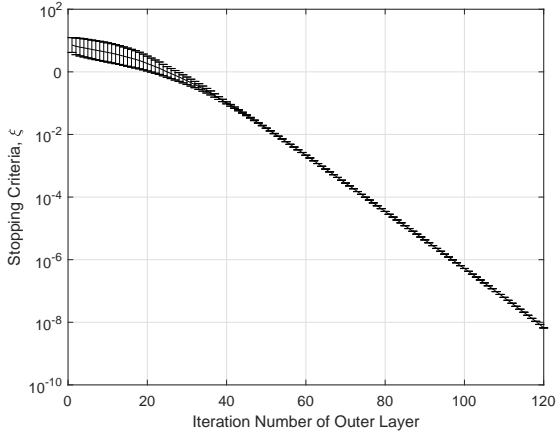}
 \caption{\small{Stopping indicator of the penalty-based algorithm.}}\label{fig_conv_penalty1}
\end{centering}
\vspace{-0.4cm}
\end{figure}

\begin{figure}[t]
\begin{centering}
\includegraphics[width=.45\textwidth]{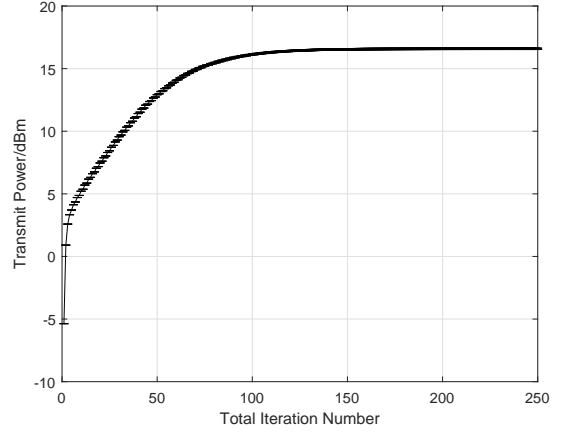}
 \caption{\small{Convergence of the penalty-based algorithm.}}\label{fig_conv_penalty2}
\end{centering}
\vspace{-0.2cm}
\end{figure}
\subsection{Convergence Performance of the Penalty-based Algorithm}
First, let us look at the convergence performance of the penalty-based algorithm. We show the stopping indicator \eqref{stop_criteria} of the penalty-based algorithm in Fig.~\ref{fig_conv_penalty1} and the average convergence of the penalty-based algorithm in Fig.~\ref{fig_conv_penalty2} in the case of continuous phase shifts of analog beamformer and RIS coefficients. These curves are plotted with the average plus and minus the standard deviation. {Note that the transmit power increases as the total number of iterations increases. This is because that a larger $\rho$ corresponding to a larger penalty for violating the equality restrictions, necessitating a higher transmit power to reduce the penalty term.} It is observed that the stopping indicator can always meet the predefined accuracy $10^{-7}$ after about 110 outer layer iterations in Fig.~\ref{fig_conv_penalty1}. Thus, the solutions obtained by \textit{Algorithm} \ref{alg_penalty} satisfy all SINR constraints. Fig.~\ref{fig_conv_penalty2} shows that the proposed algorithm converges after about 200 total iterations, which means that the inner layer runs averagely 2 times.

\subsection{Performance and Computational Comparison of Solving Problem \eqref{optimy1} by Different Methods}
\label{simulation_methods}
\begin{figure}[t]
\begin{centering}
\includegraphics[width=.45\textwidth]{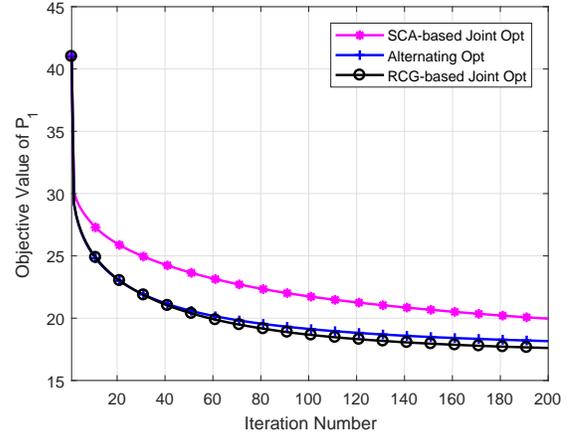}
 \caption{{Convergence comparison with fixed penalty $\rho=1$ when solving problem \eqref{optimy1} by different methods.}}\label{conv1}
\end{centering}
\end{figure}

\begin{figure}[t]
\begin{centering}
\includegraphics[width=.45\textwidth]{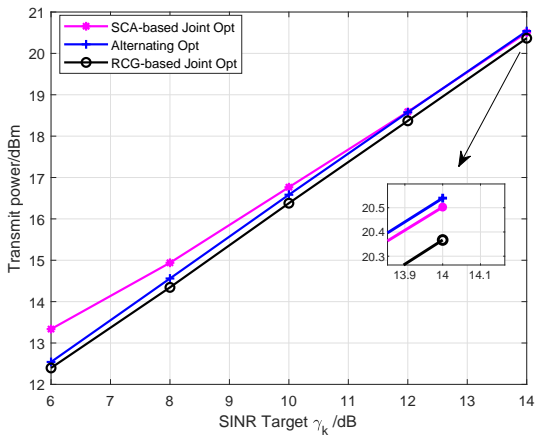}
 \caption{{Transmit power versus SINR targets when solving problem \eqref{optimy1} by different methods.}}\label{conv2}
\end{centering}
\vspace{-0.2cm}
\end{figure}
{We first compare the performance of different methods of solving problem \eqref{optimy1} as described in Section \ref{manifold_b}. Fig.~\ref{conv1} illustrates the objective value of $\mathcal{P}_1$ versus the iteration number when the penalty factor $\rho$ is fixed to one. Fig.~\ref{conv2} illustrates the transmit power versus SINR targets.
Though the optimal solution can be obtained for each subproblem in alternating optimization, it converges to a worse local optimum compared with the RCG-based joint optimization  as shown in Fig.~\ref{conv1} and Fig.~\ref{conv2}. Though the SCA-based joint optimization does not require projection, it performs worse than the RCG-based joint optimization as shown in Fig.~\ref{conv1} and Fig.~\ref{conv2}. It is also seen from Fig.~\ref{conv2} that the gap between the SCA-based method and the RCG-based method decreases as the SINR targets increase. However, we have tested the results when the SINR target is 20dB, the RCG-based method still outperforms the SCA-based method. }

{We further compare the computational time with fixed penalty $\rho=1$ when solving problem (15) by different methods in Table \ref{compu_comp2}. Here, we set the RIS $F_1 \times F_2$ unit cells where $F_1=5$ and $F_2$ can vary. It is found that the SCA-based Joint Opt runs the fastest, while the Alternating Opt runs the slowest.}

{Overall, the RCG-based Joint Opt converges to the best point, and the time consumed is somewhere in the middle. Therefore, the RCG-based Joint Opt is a good choice among the three methods. In the following, we adopt the RCG-based joint optimization method.}

\begin{table}
\centering
\begin{tabular}{|c|c|c|c|c|}
\hline \multirow{2}{*} {} & \multicolumn{4}{|c|} {\text { Running time (s) }} \\
\cline { 2 - 5 } & F=10 & F=20 & F=40 & F=80 \\
\hline \text { Alternating Opt } & 142.9865 & 152.4747 & 157.7086 & 163.2711 \\
\hline \text { RCG-based Joint Opt } & 134.5921 & 142.8421 & 143.1132 & 148.7838 \\
\hline \text { SCA-based Joint Opt } & 103.1461 & 104.0654 & 105.5587 & 110.4485 \\
\hline
\end{tabular}
\caption{Computational Time Comparison.} \label{compu_comp2}
\end{table}

\subsection{Influence of Discrete Phase Shifts}
\begin{figure}[t]
\begin{centering}
\includegraphics[width=.45\textwidth]{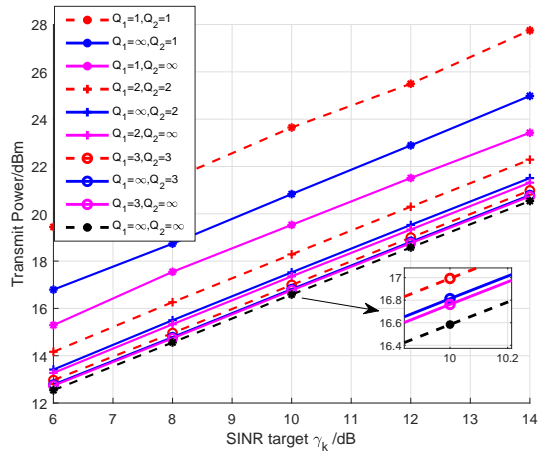}
 \caption{Influence of discrete phase shifts.}\label{fig_discrete1}
\end{centering}
\vspace{-0.5cm}
\end{figure}
We consider that the number of control bits at the analog beamformer and at the RIS, i.e., $Q_1$ and $Q_2$, can be designed separately, and each can take values from $\{1, 2, 3, \infty\}$, where $\infty$ corresponds to continuous phase shifts. Fig.~\ref{fig_discrete1} shows that when there is only one control bit for both analog beamformer and RIS, i.e. $Q_1 = Q_2 = 1$, the power gap to the ideal case with continuous phase shifts is up to 7 dB; when $Q_1 = Q_2 = 2$ and $Q_1 = Q_2 = 3$, the gap reduces quickly to 1.5 dB and 0.4dB, respectively. This suggests that having 3 bits for the discrete phase shifts is enough in practice. It is also seen from Fig.~\ref{fig_discrete1} that the BS is more robust to the discrete phase shifts than the RIS. In specific, the performance at $Q_1=1, Q_2 = \infty$ is about 2 dB better than that at $Q_1 = \infty, Q_2= 1$.
{We believe that the analog beamforming at the BS has a larger dimension of regulation than the RIS. Specifically, the analog beamforming contains many RF chains and each RF chain can serve one user, while all users are served by the same RIS. Therefore, the BS is more robust to the discrete phase shifts than the RIS.
}

\subsection{Performance Comparison with Other Schemes}
To demonstrate the efficiency of the proposed algorithms and to reveal some design insights, we compare the performance of the following algorithms when $Q_{1}=3$ and $Q_2=3$.
\begin{itemize}
\item Penalty-Manifold joint design with hybrid beamforming structure (Penalty-Manifold HB):  This is the proposed \textit{Algorithm} \ref{alg_penalty} for joint design of hybrid beamforming and RIS phase shifts.

\item Penalty-Manifold joint design with fully digital beamforming structure (Penalty-Manifold FD): This is the proposed \textit{Algorithm} \ref{alg_penalty} but changing the hybrid beamforming to the fully digital beamforming at the BS. This is done by setting $D=1$.

\item Penalty-Manifold joint design with random $\mathbf{\Theta}$ (Random $\mathbf{\Theta}$):  The phase shifts at the RIS are randomly selected to be feasible values. Then the hybrid beamforming matrices $\{\mathbf{W},\mathbf{V}\}$ at the BS are obtained by using the penalty-manifold joint algorithm as in \textit{Algorithm} \ref{alg_penalty}, where the update of $\mathbf{\Theta}$ is skipped. This is to find out the significance of optimizing the phase shifts at the RIS.

\item Penalty-Manifold joint design with SDR $\mathbf{\Theta}$ (SDR $\mathbf{\Theta}$):  The phase shifts at the RIS are designed by using the SDR approach as stated in Section \ref{ris_design}. Then the hybrid beamforming matrices $\{\mathbf{W}, \mathbf{V}\}$ at the BS are obtained by using the penalty-manifold joint algorithm as in \textit{Algorithm} \ref{alg_penalty}, where the udpate of $\mathbf{\Theta}$ is skipped. This is again to find out the significance of optimizing the phase shifts at the RIS.

\item BCD-SDR joint design (BCD-SDR):  The conventional BCD method in conjunction with the SDR method, as mentioned in Section \ref{pro_formulation}.

\item  Sequential design:  the proposed sequential design where RIS phase shifts, analog beamforming, and digital beamforming are optimized sequentially in Section \ref{individual}. {In order to make the sequential optimization method be more effective, we try different overlapping coefficients $\mu$ from 1 to 4 and let the best result be the final solution.}
\end{itemize}

\begin{figure}[t]
\begin{centering}
\includegraphics[width=.45\textwidth]{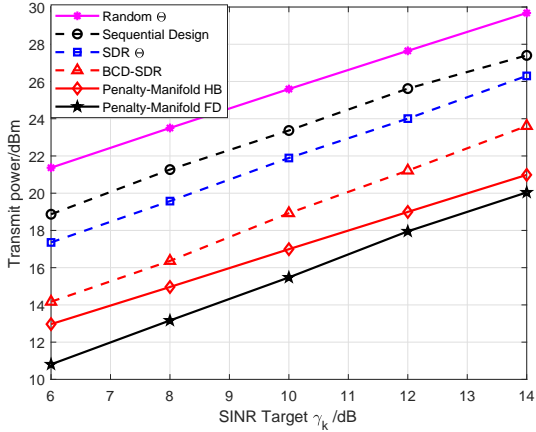}
 \caption{\small{Transmit power versus SINR targets.}}\label{fig_multi_sinr}
\end{centering}
\end{figure}

\begin{figure}[t]
\begin{centering}
\includegraphics[width=.45\textwidth]{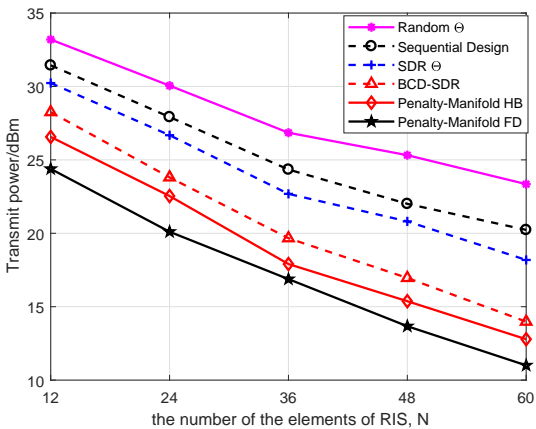}
 \caption{\small{Transmit power versus the number of the elements of RIS.}} \label{fig_ele}
\end{centering}
\end{figure}

\begin{figure}[t]
\begin{centering}
\includegraphics[width=.45\textwidth]{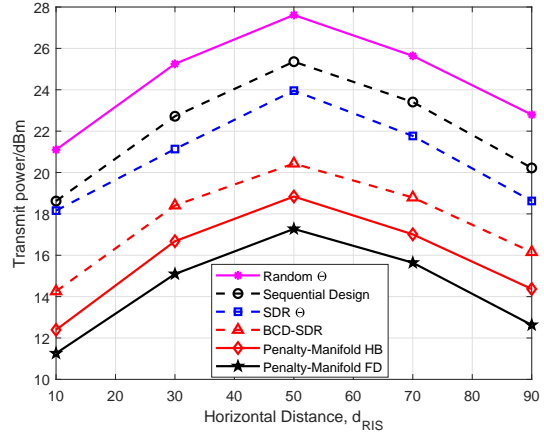}
 \caption{\small{Transmit power versus the horizontal distance of RIS.}}\label{fig_distance}
\end{centering}
\end{figure}

\begin{table}[t]
\centering
\begin{tabular}{|c|c|c|c|c|}
\hline \multirow{2}{*} {} & \multicolumn{4}{|c|} {\text { Running time (s) }} \\
\cline { 2 - 5 } & F=10 & F=20 & F=40 & F=80 \\
\hline \text { SDR-BCD } & 54.2175 & 61.1350 & 169.0588 & 461.3819 \\
\hline \text { Penalty-Manifold FD } & 96.0028 & 101.2406 & 115.0541 & 116.8831 \\
\hline \text { Sequential Design } & 15.3721 & 17.9422 & 20.1504 & 37.0946 \\
\hline
\end{tabular}
\caption{Computational Time Comparison.} \label{compu_comp}
\end{table}

Fig.~\ref{fig_multi_sinr} illustrates the transmit power versus SINR targets. We first observe that the Penalty-Manifold joint design outperforms the start-of-the-art BCD-SDR joint design, which verifies the effectiveness of the proposed algorithm. Second, it is seen that the Penalty-Manifold joint design with random $\mathbf{\Theta}$ performs the worst among all the considered schemes. By simply changing the random $\mathbf{\Theta}$ to the SDR $\mathbf{\Theta}$ (while keeping the joint design of $\{\mathbf{W},\mathbf{V}\}$ unchanged), the transmit power consumption can be reduced by 4 dB. If $\mathbf{\Theta}$ is involved in the Penalty-Manifold joint design, another  about 5 dB power reduction can be obtained. These observations indicate that the design of RIS phase shifts plays the crucial role for performance optimization.
Third, we observe that the sequential design is about 1dB worse than the joint design with SDR $\mathbf{\Theta}$. This suggests that, when the RIS response matrix is designed sequentially, further optimizing the hybrid beamforming at the BS can only bring marginal improvement.
Last but not least, we observe that the power consumed by Penalty-Manifold beamforming is about 2dB higher than the power consumed by Penalty-Manifold FD. Note that the hybrid beamforming has a much lower hardware cost since it only employs $N=6$ RF chains at the BS, while the fully digital beamforming has $M=36$ RF chains. This means that the proposed hybrid beamforming is effective.

The influence of the RIS element number is considered in Fig.~\ref{fig_ele}. When the RIS element number increases from 12 to 60, the transmit power decreases about 15dB. Thus, we conclude that the RIS can greatly reduce the transmit power by installing a large number of elements.

Fig.~\ref{fig_distance} illustrates the transmit power versus the RIS horizontal distance. It is seen  that as the RIS horizontal distance $d_{RIS}$ increases, the transmit power increases firstly, and reaches the peak at 50 m, then decreases. This can be explained that the received power through the reflection of the RIS in the far field is proportional to $d_1^{-2} d_2^{-2}$, where $d_1$ and $d_2$ denote the distances between the BS-RIS and RIS-user, respectively. It is found that the RIS can be located near the BS or users to save energy.

{We further compare in Table \ref{compu_comp} the running time for various values of $F$. Here, $\mu$ is set to be 3. We set the RIS $F_1 \times F_2$ unit cells where $F_1=5$ and $F_2$ can vary. It is observed that the time consumed by the SDR-BCD method increases greatly as $F$ increases. It is interesting that the computational time of Penalty-Manifold FD is insensitive to $F$. And the time consumed by the Sequential Design is the least among the algorithms, which means that it has the lowest complexity.}

\section{Conclusion}
\label{sec_conclusion}
{In this paper, we investigate an RIS-aied downlink MIMO system, with the objective of minimizing the transmit power at the BS by jointly optimizing the hybrid A/D beamforming at the BS, as well as the overall response-coefficient at the RIS, subject to individual minimum SINR constraints. The non-convex problem is first solved by the penalty-based algorithm with manifold optimization, followed by a low-complexity sequential optimization. {In particular, we propose three different methods for optimizing the BS analog beamforming and the RIS response matrix in the penalty-based algorithm. The RCG-based joint optimization is found to outperform the other two methods but it has a slightly higher complexity.} Extensive simulation results demonstrate that the proposed algorithm outperforms the state-of-art BCD-SDR algorithm. Our simulation results provide useful insights into the corresponding wireless system design. In particular, the simulation results show that utilizing a large number of RIS units could help reduce the transmit power at the BS greatly. Moreover, 3-bit quantizers of both the RIS and the analog beamformer could approach the performance of continuous phase shifters.}

\bibliographystyle{IEEEtran}
\bibliography{hybrid3_2}

\begin{thebibliography}{10}
\providecommand{\url}[1]{#1}
\csname url@samestyle\endcsname
\providecommand{\newblock}{\relax}
\providecommand{\bibinfo}[2]{#2}
\providecommand{\BIBentrySTDinterwordspacing}{\spaceskip=0pt\relax}
\providecommand{\BIBentryALTinterwordstretchfactor}{4}
\providecommand{\BIBentryALTinterwordspacing}{\spaceskip=\fontdimen2\font plus
\BIBentryALTinterwordstretchfactor\fontdimen3\font minus
  \fontdimen4\font\relax}
\providecommand{\BIBforeignlanguage}[2]{{%
\expandafter\ifx\csname l@#1\endcsname\relax
\typeout{** WARNING: IEEEtran.bst: No hyphenation pattern has been}%
\typeout{** loaded for the language `#1'. Using the pattern for}%
\typeout{** the default language instead.}%
\else
\language=\csname l@#1\endcsname
\fi
#2}}
\providecommand{\BIBdecl}{\relax}
\BIBdecl

\bibitem{9417417}
B.~Guo, R.~Li, and M.~Tao, ``Joint design of hybrid beamforming and phase
  shifts in {RIS}-aided mmwave communication systems,'' in \emph{Proc. IEEE
  Wireless Commun. Netw. Conf. (WCNC)}, Mar. 2021, pp. 1--6.

\bibitem{6732923}
S.~{Rangan}, T.~S. {Rappaport}, and E.~{Erkip}, ``Millimeter-wave cellular
  wireless networks: Potentials and challenges,'' \emph{Proc. IEEE}, vol. 102,
  no.~3, pp. 366--385, Mar. 2014.

\bibitem{6824746}
A.~{Ghosh}, T.~A. {Thomas}, M.~C. {Cudak}, R.~{Ratasuk}, P.~{Moorut}, F.~W.
  {Vook}, T.~S. {Rappaport}, G.~R. {MacCartney}, S.~{Sun}, and S.~{Nie},
  ``Millimeter-wave enhanced local area systems: A high-data-rate approach for
  future wireless networks,'' \emph{IEEE J. Sel. Areas Commun.}, vol.~32,
  no.~6, pp. 1152--1163, Jun. 2014.

\bibitem{niu2015survey}
Y.~Niu, Y.~Li, D.~Jin, L.~Su, and A.~V. Vasilakos, ``A survey of millimeter
  wave communications ({mmWave}) for {5G}: Opportunities and challenges,''
  \emph{Wireless Netw}, vol.~21, no.~8, pp. 2657--2676, Apr. 2015.

\bibitem{book}
T.~Cui, D.~Smith, and R.~Liu, \emph{Metamaterials: Theory, Design, and
  Applications}.\hskip 1em plus 0.5em minus 0.4em\relax Springer, 2010.

\bibitem{8910627}
Q.~{Wu} and R.~{Zhang}, ``Towards smart and reconfigurable environment:
  Intelligent reflecting surface aided wireless network,'' \emph{IEEE Commun.
  Mag.}, vol.~58, no.~1, pp. 106--112, Jan. 2020.

\bibitem{8796365}
E.~{Basar}, M.~{Di Renzo}, J.~{De Rosny}, M.~{Debbah}, M.~{Alouini}, and
  R.~{Zhang}, ``Wireless communications through reconfigurable intelligent
  surfaces,'' \emph{IEEE Access}, vol.~7, pp. 116\,753--116\,773, Aug. 2019.

\bibitem{9122596}
S.~{Gong}, X.~{Lu}, D.~T. {Hoang}, D.~{Niyato}, L.~{Shu}, D.~I. {Kim}, and
  Y.~C. {Liang}, ``Toward smart wireless communications via intelligent
  reflecting surfaces: A contemporary survey,'' \emph{IEEE Commun. Surv.
  Tutor.}, vol.~22, no.~4, pp. 2283--2314, Jun. 2020.

\bibitem{9086766}
M.~A. {ElMossallamy}, H.~{Zhang}, L.~{Song}, K.~G. {Seddik}, Z.~{Han}, and
  G.~Y. {Li}, ``Reconfigurable intelligent surfaces for wireless
  communications: Principles, challenges, and opportunities,'' \emph{IEEE
  Trans. Cogn. Commun. Netw.}, vol.~6, no.~3, pp. 990--1002, Sep. 2020.

\bibitem{wu2019intelligent}
Q.~Wu and R.~Zhang, ``Intelligent reflecting surface enhanced wireless network
  via joint active and passive beamforming,'' \emph{IEEE Trans. Wireless
  Commun.}, vol.~18, no.~11, pp. 5394--5409, Nov. 2019.

\bibitem{9226616}
P.~{Wang}, J.~{Fang}, X.~{Yuan}, Z.~{Chen}, and H.~{Li}, ``Intelligent
  reflecting surface-assisted millimeter wave communications: Joint active and
  passive precoding design,'' \emph{IEEE Trans. Veh. Technol.}, pp. 1--1, Dec.
  2020.

\bibitem{9246254}
H.~{Xie}, J.~{Xu}, and Y.-F. {Liu}, ``Max-min fairness in {IRS}-aided
  multi-cell {MISO} systems with joint transmit and reflective beamforming,''
  \emph{IEEE Trans. Wireless Commun.}, vol.~20, no.~2, pp. 1379--1393, Feb.
  2021.

\bibitem{8741198}
C.~{Huang}, A.~{Zappone}, G.~C. {Alexandropoulos}, M.~{Debbah}, and C.~{Yuen},
  ``Reconfigurable intelligent surfaces for energy efficiency in wireless
  communication,'' \emph{IEEE Trans. Wireless Commun.}, vol.~18, no.~8, pp.
  4157--4170, Aug. 2019.

\bibitem{li2019joint}
\BIBentryALTinterwordspacing
X.~Li, J.~Fang, F.~Gao, and H.~Li, ``Joint active and passive beamforming for
  intelligent reflecting surface-assisted massive {MIMO} systems,'' 2019.
  [Online]. Available: \url{https://arxiv.org/abs/1912.00728}
\BIBentrySTDinterwordspacing

\bibitem{guo2020weighted}
H.~Guo, Y.-C. Liang, J.~Chen, and E.~G. Larsson, ``Weighted sum-rate
  maximization for reconfigurable intelligent surface aided wireless
  networks,'' \emph{IEEE Trans. Wireless Commun.}, vol.~19, no.~5, pp.
  3064--3076, May 2020.

\bibitem{8723525}
M.~{Cui}, G.~{Zhang}, and R.~{Zhang}, ``Secure wireless communication via
  intelligent reflecting surface,'' \emph{IEEE Wireless Commun. Lett.}, vol.~8,
  no.~5, pp. 1410--1414, Oct. 2019.

\bibitem{9198898}
A.~{Almohamad}, A.~M. {Tahir}, A.~{Al-Kababji}, H.~M. {Furqan}, T.~{Khattab},
  M.~O. {Hasna}, and H.~{Arslan}, ``Smart and secure wireless communications
  via reflecting intelligent surfaces: A short survey,'' \emph{IEEE Open J.
  Commun. Soc.}, vol.~1, pp. 1442--1456, Sep. 2020.

\bibitem{8959174}
S.~{Li}, B.~{Duo}, X.~{Yuan}, Y.~{Liang}, and M.~{Di Renzo}, ``Reconfigurable
  intelligent surface assisted {UAV} communication: Joint trajectory design and
  passive beamforming,'' \emph{IEEE Wireless Commun. Lett.}, vol.~9, no.~5, pp.
  716--720, May 2020.

\bibitem{9124704}
L.~{Yang}, F.~{Meng}, J.~{Zhang}, M.~O. {Hasna}, and M.~D. {Renzo}, ``On the
  performance of ris-assisted dual-hop {UAV} communication systems,''
  \emph{IEEE Commun. Surv. Tutor.}, vol.~69, no.~9, pp. 10\,385--10\,390, Sep.
  2020.

\bibitem{8941080}
Q.~{Wu} and R.~{Zhang}, ``Weighted sum power maximization for intelligent
  reflecting surface aided {SWIPT},'' \emph{IEEE Wireless Commun. Lett.},
  vol.~9, no.~5, pp. 586--590, May 2020.

\bibitem{wu2020joint}
Q.~Wu and R.~Zhang, ``Joint active and passive beamforming optimization for
  intelligent reflecting surface assisted {SWIPT} under {QoS} constraints,''
  \emph{IEEE J. Sel. Areas Commun.}, vol.~38, no.~8, pp. 1735--1748, Aug. 2020.

\bibitem{8030501}
A.~F. {Molisch}, V.~V. {Ratnam}, S.~{Han}, Z.~{Li}, S.~L.~H. {Nguyen}, L.~{Li},
  and K.~{Haneda}, ``Hybrid beamforming for massive {MIMO}: A survey,''
  \emph{IEEE Commun. Mag.}, vol.~55, no.~9, pp. 134--141, Sep. 2017.

\bibitem{7389996}
F.~{Sohrabi} and W.~{Yu}, ``Hybrid digital and analog beamforming design for
  large-scale antenna arrays,'' \emph{IEEE J. Sel. Topics Signal Process},
  vol.~10, no.~3, pp. 501--513, Apr. 2016.

\bibitem{ying2020gmdbased}
K.~{Ying}, Z.~{Gao}, S.~{Lyu}, Y.~{Wu}, H.~{Wang}, and M.~{Alouini},
  ``{GMD}-based hybrid beamforming for large reconfigurable intelligent surface
  assisted millimeter-wave massive {MIMO},'' \emph{IEEE Access}, vol.~8, pp.
  19\,530--19\,539, Jan. 2020.

\bibitem{xiu2020reconfigurable}
\BIBentryALTinterwordspacing
Y.~Xiu, J.~Zhao, W.~Sun, M.~D. Renzo, G.~Gui, Z.~Zhang, and N.~Wei,
  ``Reconfigurable intelligent surfaces aided {mmWave} {NOMA}: Joint power
  allocation, phase shifts, and hybrid beamforming optimization,'' 2020.
  [Online]. Available: \url{https://arxiv.org/abs/2007.05873}
\BIBentrySTDinterwordspacing

\bibitem{9234098}
P.~{Wang}, J.~{Fang}, L.~{Dai}, and H.~{Li}, ``Joint transceiver and large
  intelligent surface design for massive {MIMO} {MmWave} systems,'' \emph{IEEE
  Trans. Wireless Commun.}, vol.~20, no.~2, pp. 1052--1064, Feb. 2021.

\bibitem{chen2019channel}
\BIBentryALTinterwordspacing
J.~Chen, Y.-C. Liang, H.~V. Cheng, and W.~Yu, ``Channel estimation for
  reconfigurable intelligent surface aided multi-user mimo systems,'' 2019.
  [Online]. Available: \url{https://arxiv.org/abs/1912.03619}
\BIBentrySTDinterwordspacing

\bibitem{9103231}
P.~Wang, J.~Fang, H.~Duan, and H.~Li, ``Compressed channel estimation for
  intelligent reflecting surface-assisted millimeter wave systems,'' \emph{IEEE
  Signal Process. Lett.}, vol.~27, pp. 905--909, May 2020.

\bibitem{9127834}
S.~Liu, Z.~Gao, J.~Zhang, M.~D. Renzo, and M.-S. Alouini, ``Deep denoising
  neural network assisted compressive channel estimation for mmwave intelligent
  reflecting surfaces,'' \emph{IEEE Trans. Veh. Technol.}, vol.~69, no.~8, pp.
  9223--9228, Aug. 2020.

\bibitem{9374451}
Z.~Wan, Z.~Gao, F.~Gao, M.~D. Renzo, and M.-S. Alouini, ``Terahertz massive
  {MIMO} with holographic reconfigurable intelligent surfaces,'' \emph{IEEE
  Trans. Commun.}, vol.~69, no.~7, pp. 4732--4750, Jul. 2021.

\bibitem{9398559}
J.~He, H.~Wymeersch, and M.~Juntti, ``Channel estimation for {RIS}-aided mmwave
  {MIMO} systems via atomic norm minimization,'' \emph{IEEE Trans. Wireless
  Commun.}, vol.~20, no.~9, pp. 5786--5797, Sep. 2021.

\bibitem{6717211}
O.~E. {Ayach}, S.~{Rajagopal}, S.~{Abu-Surra}, Z.~{Pi}, and R.~W. {Heath},
  ``Spatially sparse precoding in millimeter wave {MIMO} systems,'' \emph{IEEE
  Trans. Wireless Commun.}, vol.~13, no.~3, pp. 1499--1513, Mar. 2014.

\bibitem{9148827}
H.~{Han}, J.~{Zhao}, D.~{Niyato}, M.~D. {Renzo}, and Q.~{Pham}, ``Intelligent
  reflecting surface aided network: Power control for physical-layer
  broadcasting,'' in \emph{Proc. IEEE Int. Conf. Commun. (ICC)}, Jun. 2020, pp.
  1--7.

\bibitem{8930608}
Q.~{Wu} and R.~{Zhang}, ``Beamforming optimization for wireless network aided
  by intelligent reflecting surface with discrete phase shifts,'' \emph{IEEE
  Trans. Commun.}, vol.~68, no.~3, pp. 1838--1851, Mar. 2020.

\bibitem{9133142}
C.~{You}, B.~{Zheng}, and R.~{Zhang}, ``Channel estimation and passive
  beamforming for intelligent reflecting surface: Discrete phase shift and
  progressive refinement,'' \emph{IEEE J. Sel. Areas Commun.}, vol.~38, no.~11,
  pp. 2604--2620, Nov. 2020.

\bibitem{absil2009optimization}
P.-A. Absil, R.~Mahony, and R.~Sepulchre, \emph{Optimization Algorithms on
  Matrix Manifolds}.\hskip 1em plus 0.5em minus 0.4em\relax Princeton
  University Press, 2009.

\bibitem{yu2016alternating}
X.~Yu, J.-C. Shen, J.~Zhang, and K.~B. Letaief, ``Alternating minimization
  algorithms for hybrid precoding in millimeter wave {MIMO} systems,''
  \emph{IEEE J. Sel. Topics Signal Process.}, vol.~10, no.~3, pp. 485--500,
  Apr. 2016.

\bibitem{yu2019miso}
\BIBentryALTinterwordspacing
X.~Yu, D.~Xu, and R.~Schober, ``{MISO} wireless communication systems via
  intelligent reflecting surfaces,'' 2019. [Online]. Available:
  \url{https://arxiv.org/abs/1904.12199}
\BIBentrySTDinterwordspacing

\bibitem{razaviyayn2013unified}
M.~Razaviyayn, M.~Hong, and Z.-Q. Luo, ``A unified convergence analysis of
  block successive minimization methods for nonsmooth optimization,''
  \emph{SIAM J. Optim.}, vol.~23, no.~2, pp. 1126--1153, 2013.

\bibitem{1561584}
A.~{Wiesel}, Y.~C. {Eldar}, and S.~{Shamai}, ``Linear precoding via conic
  optimization for fixed {MIMO} receivers,'' \emph{IEEE Trans. Signal
  Process.}, vol.~54, no.~1, pp. 161--176, Jan. 2006.

\bibitem{7558213}
Q.~Shi, M.~Hong, X.~Gao, E.~Song, Y.~Cai, and W.~Xu, ``Joint source-relay
  design for full-duplex mimo af relay systems,'' \emph{IEEE Trans. Signal
  Process.}, vol.~64, no.~23, pp. 6118--6131, Dec. 2016.

\bibitem{grant2014cvx}
\BIBentryALTinterwordspacing
M.~Grant and S.~Boyd, ``{CVX}: Matlab software for disciplined convex
  programming, version 2.1,'' 2014. [Online]. Available:
  \url{http://cvxr.com/cvx}
\BIBentrySTDinterwordspacing

\bibitem{5447068}
Z.-Q. {Luo}, W.-K. {Ma}, A.~M. {So}, Y.~{Ye}, and S.~{Zhang}, ``Semidefinite
  relaxation of quadratic optimization problems,'' \emph{IEEE Signal Process.
  Mag.}, vol.~27, no.~3, pp. 20--34, May 2010.

\bibitem{8362957}
J.~{Mao}, Z.~{Gao}, Y.~{Wu}, and M.~{Alouini}, ``Over-sampling codebook-based
  hybrid minimum sum-mean-square-error precoding for millimeter-wave
  {3D-MIMO},'' \emph{IEEE Wireless Commun. Lett.}, vol.~7, no.~6, pp. 938--941,
  Dec. 2018.

\bibitem{ben2001lectures}
A.~Ben-Tal and A.~Nemirovski, \emph{Lectures on modern convex optimization:
  analysis, algorithms, and engineering applications}.\hskip 1em plus 0.5em
  minus 0.4em\relax SIAM, 2001.

\bibitem{7961152}
K.~Venugopal, A.~Alkhateeb, N.~González~Prelcic, and R.~W. Heath, ``Channel
  estimation for hybrid architecture-based wideband millimeter wave systems,''
  \emph{IEEE J. Sel. Areas Commun.}, vol.~35, no.~9, pp. 1996--2009, Sep. 2017.

\bibitem{1634819}
N.~D. {Sidiropoulos}, T.~N. {Davidson}, and {Zhi-Quan Luo}, ``Transmit
  beamforming for physical-layer multicasting,'' \emph{IEEE Trans. Signal
  Process.}, vol.~54, no.~6, pp. 2239--2251, Jun. 2006.

\bibitem{4443878}
E.~{Karipidis}, N.~D. {Sidiropoulos}, and Z.~{Luo}, ``Quality of service and
  max-min fair transmit beamforming to multiple cochannel multicast groups,''
  \emph{IEEE Trans. Signal Process.}, vol.~56, no.~3, pp. 1268--1279, Mar.
  2008.

\bibitem{akdeniz2014millimeter}
M.~R. Akdeniz, Y.~Liu, M.~K. Samimi, S.~Sun, S.~Rangan, T.~S. Rappaport, and
  E.~Erkip, ``Millimeter wave channel modeling and cellular capacity
  evaluation,'' \emph{IEEE J. Sel. Areas Commun.}, vol.~32, no.~6, pp.
  1164--1179, Jun. 2014.

\end{thebibliography}

\end{document}